\documentclass[ejs,preprint]{imsart}

\RequirePackage[OT1]{fontenc}
\RequirePackage{amsthm,amsmath}
\RequirePackage[colorlinks,citecolor=blue,urlcolor=blue]{hyperref}
\RequirePackage[numbers]{natbib}
\setattribute{journal}{name}{}

\usepackage{amsthm,amsmath,natbib}
\RequirePackage[colorlinks,citecolor=blue,urlcolor=blue]{hyperref}

\usepackage{amssymb}
\usepackage{bm}
\usepackage{mathrsfs, color}
\usepackage{subfigure}
\usepackage{graphicx}
\usepackage{graphics}
\usepackage{booktabs}
\usepackage{natbib}
\usepackage{url}
\usepackage{enumerate}
\usepackage{bbm}

\theoremstyle{plain}% default
\newtheorem{thm}{Theorem}
\newtheorem{lem}{Lemma}

\newtheorem{cor}[thm]{Corollary}

\theoremstyle{definition}

\theoremstyle{remark}

\allowdisplaybreaks

\begin{document}

\begin{frontmatter}

% "Title of the paper"
\title{Variable Screening for High Dimensional Time Series}
\runtitle{Screening for Time Series}

% indicate corresponding author with \corref{}
\begin{aug}
%\author{\fnms{Kashif} \snm{Yousuf}\thanksref{m1}\ead[label=e1]{ky2304@columbia.edu}},
\author{\fnms{Kashif} \snm{Yousuf}\ead[label=e1]{ky2304@columbia.edu}},
\runauthor{Kashif Yousuf}

\affiliation{Columbia University\thanksmark{m1}}

\address{Columbia University \\ Department of Statistics \\
1255 Amsterdam Ave \\
New York, NY 10027 \\
\printead{e1}}
\end{aug}

\runauthor{K. Yousuf}

\begin{abstract}
Variable selection is a widely studied problem in high dimensional statistics, primarily since estimating the precise relationship between the covariates and the response is of great importance in many scientific disciplines. However, most of theory and methods developed towards this goal for the linear model invoke the assumption of iid sub-Gaussian covariates and errors. This paper analyzes the theoretical properties of Sure Independence Screening (SIS) (Fan and Lv \cite{FanLv2008}) for high dimensional linear models with dependent and/or heavy tailed covariates and errors. We also introduce a generalized least squares screening (GLSS) procedure which utilizes the serial correlation present in the data. By utilizing this serial correlation when estimating our marginal effects, GLSS is shown to outperform SIS in many cases. For both procedures we prove sure screening properties, which depend on the moment conditions, and the strength of dependence in the error and covariate processes, amongst other factors. Additionally, combining these screening procedures with the adaptive Lasso is analyzed. Dependence is quantified by functional dependence measures (Wu \cite{Wu2005}), and the results rely on the use of Nagaev-type and exponential inequalities for dependent random variables. We also conduct simulations to demonstrate the finite sample performance of these procedures, and include a real data application of forecasting the US inflation rate.

\end{abstract}



\end{frontmatter}

\section{Introduction} \label{sec:introduction}

With the advancement of data acquisition technology, high dimensionality is a characteristic of data being collected in fields as diverse as health sciences, genomics, neuroscience, astronomy, finance, and macroeconomics. Applications where we have a large number of predictors for a relatively small number of observations are becoming increasingly common. For example, in disease classification we usually have thousands of variables, such as expression of genes, which are collected, while the sample size is usually in the tens. Other examples include fMRI data, where the number of voxels can number in the thousands and far outnumber the observations. For an overview of high dimensionality in economics and finance, see \cite{Fanetal2011}. For the biological sciences, see \cite{FanRen2006,Bickeletal2009} and references therein. The main goals in these situations according to \cite{Bickel2008} are:
\begin{itemize}
\item To construct as effective a method as possible to predict future observations.
\item To gain insight into the relationship between features and response for scientific purposes, as well as hopefully, to construct an improved prediction method.
\end{itemize}

More formally we are dealing with the case where 
\begin{equation}\bm{y}=\bm{X}\bm{\beta}+\bm{\epsilon} \label{eq1}\end{equation} 
with $\bm{y} =\left(Y_1,\ldots,Y_n\right)^T$ being an $n$-vector of responses, $\bm{X}=\left(\bm{x_1},\ldots,\bm{x_n}\right)^T$ being an $n\times p_n$ random design matrix, and $\bm{\epsilon}=\left(\epsilon_1,\ldots,\epsilon_n\right)^T$ is a random vector of errors. In addition, when the dimensionality of the predictors ($p_n$) is large we usually make the assumption that the underlying coefficient vector $(\bm{\beta})$ is sparse. Sparsity is a characteristic that is frequently found in many scientific applications \cite{FanLv2008},\cite{Johnstone2009}. For example, in disease classification it is usually the case that only a small amount of genes are relevant to predicting the outcome.

Indeed, there are a wealth of theoretical results and methods that are devoted to this issue. Our primary focus is on screening procedures. Sure Independence Screening (SIS) as originally introduced in \cite{FanLv2008}, was applicable to the linear model, and is based on a ranking of the absolute values of the marginal correlations of the predictors with the response. This method allows one to deal with situations in which the number of predictors is of an exponential order of the number of observations, which they termed as ultrahigh dimensionality. Further work on the topic has expanded the procedure to cover the case of generalized linear models \cite{FanSong2010}, non-parametric additive models \cite{Feng2011}, Cox proportional hazards model \cite{fan2010high}, single index hazard rate models \cite{Gorst2013}, and varying coefficient models \cite{Fanetal2014}. Model-free screening methods have also been developed. For example; screening using distance correlation was analyzed in \cite{Lietal2012}, a martingale difference correlation approach was introduced in \cite{Shao2014}, additional works include \cite{Zhuetal2011}, \cite{Huang2016} among others. For an overview of works related to screening procedures, one can consult \cite{Liuetal2015}. The main result introduced with these methods is that, under appropriate conditions, we can reduce the predictor dimension from size $p_n=O\left(\exp\left(n^\alpha\right)\right)$, for some $\alpha<1$, to a size $d_n$, while retaining all the relevant predictors with probability approaching 1. 

Another widely used class of methods is based on the penalized least squares approach. An overview of these methods is provided in \cite{FanLv2010} and \cite{Bickeletal2009}. Examples of methods in this class are the Lasso \cite{Tibshirani1996}, and the adaptive Lasso \cite{Zou2006}. Various theoretical results have been discovered for these class of methods. They broadly fall into analyzing the prediction error $|X(\hat{\beta}-\beta)|_{2}^{2}$, parameter estimation error $|\hat{\beta}-\beta|_{1}$, model selection consistency, as well as limiting distributions of the estimated parameters (see \cite{Bulhmann2011} for a comprehensive summary). Using screening procedures in conjunction with penalized least squares methods, such as the adaptive Lasso, presents a powerful tool for variable selection. Variable screening can allow us to quickly reduce the parameter dimension $p_n$ significantly, which weakens the assumptions needed for model selection consistency of the adaptive Lasso \cite{Huangetal2008, MM2016}. 

A key limitation of the results obtained for screening methods, is the assumption of independent observations. In addition, it is usually assumed that the covariates and the errors are sub-Gaussian (or sub-exponential). However, there are many examples of real world data where these assumptions are violated. Data which is observed over time and/or space such as meteorological data, longitudinal data, economic and financial time series frequently exhibit covariates and/or errors which are serially correlated. One specific example is the case of fMRI time series, where there can exist a complicated spatial-temporal dependence structure in the errors and the covariates (see \cite{Worsley2002}). Another example is in forecasting macroeconomic indicators such as GDP or inflation rate, where we have large number of macroeconomic and financial time series, along with their lags, as possible covariates. Examples of heavy tailed and dependent errors and covariates can be found most prominently in financial, insurance and macroeconomic data. 

These examples stress why it is extremely important for variable selection methods to be capable of handling scenarios where the assumption of independent sub-Gaussian (or sub-exponential) observations is violated. Some works related to this goal for the Lasso include \cite{Wangetal2007}, which extended the Lasso to jointly model the autoregressive structure of the errors as well as the covariates. However, their method is applicable only to the case where $p_n<n$, and they assume an autoregressive structure where the order of the process is known. Whereas \cite{WuandWu2016} studied the theoretical properties of the Lasso assuming a fixed design in the case of heavy tailed and dependent errors. Additionally \cite{Basu2015}, and \cite{Kock2015} investigated theoretical properties of the Lasso for high-dimensional Gaussian processes. Most recently \cite{MM2016} analyzed the adaptive Lasso for high dimensional time series while allowing for both heavy tailed covariate and errors processes, with the additional assumption that the error process is a martingale difference sequence.

Some works related to this goal for screening methods include \cite{LiGetal2012}, which allows for heavy tailed errors and covariates. Additionally \cite{Chang2013}, \cite{Shuang2014}, and \cite{Zhuetal2011} also relax the Gaussian assumption, with the first two requiring the tails of the covariates and the response to be exponentially light, while the latter allows for heavy tailed errors provided the covariates are sub-exponential. Although these works relax the moment and distributional assumptions on the covariates and the response, they still remain in the framework of independent observations. A few works have dealt with correlated observations in the context of longitudinal data (see \cite{Cheng2014},\cite{XuB2014}). However, the dependence structure of longitudinal data is too restrictive to cover the type of dependence present in most time series. Most recently \cite{Chenlu2017} proposed a non-parametric kernel smoothing screening method applicable to time series data. In their work they assume a sub-exponential response, covariates that are bounded and have a density, as well as assuming the sequence $\{(Y_i,\bm{x}_i)\}$ is strong mixing, with the additional assumption that the strong mixing coefficients decay geometrically. These assumptions can be quite restrictive; they exclude, for example, heavy tailed time series, and discrete valued time series which are common in fields such as macroeconomics, finance, neuroscience, amongst others \cite{Davis2016}. 

In this work, we study the theoretical properties of SIS for the linear model with dependent and/or heavy tailed covariates and errors. This allows us to substantially increase the number of situations in which SIS can be applied. However, one of the drawbacks to using SIS in a time series setting is that the temporal dependence structure between observations is ignored. In an attempt to correct this, we introduce a generalized least squares screening (GLSS) procedure, which utilizes this additional information when estimating the marginal effect of each covariate. By using GLS to estimate the marginal regression coefficient for each covariate, as opposed to OLS used in SIS, we correct for the effects of serial correlation. Our simulation results show the effectiveness of GLSS over SIS, is most pronounced when we have strong levels of serial correlation and weak signals. Using the adaptive Lasso as a second stage estimator after applying the above screening procedures is also analyzed. Probability bounds for our combined two stage estimator being sign consistent are provided, along with comparisons between our two stage estimator and the adaptive Lasso as a stand alone procedure.

Compared to previous work, we place no restrictions on the distribution of the covariate and error processes besides existence of a certain number of finite moments. In order to quantify dependence, we rely on the functional dependence measure framework introduced by \cite{Wu2005}, rather than the usual strong mixing coefficients. Comparisons between functional dependence measures and strong mixing assumptions are discussed in section \ref{sec:section 2}. For both GLSS and SIS, we present the sure screening properties and show the range of $p_n$ can vary from the high dimensional case, where $p_n$ is a power of $n$, to the ultrahigh dimensional case discussed in \cite{FanLv2008}. We detail how the range of $p_n$ and the sure screening properties are affected by the strength of dependence and the moment conditions of the errors and covariates, the strength of the underlying signal, and the sparsity level, amongst other factors. 

The rest of the paper is organized as follows: Section \ref{sec:section 2} reviews the functional and predictive dependence measures which will allow us to characterize the dependence in the covariate ($\bm{x_i}, i=1,...,n$) and error processes. We also discuss the assumptions placed on structure of the covariate and error processes; these assumptions are very mild, allowing us to represent a wide variety of stochastic processes which arise in practice. Section \ref{sec:section 3} presents the sure screening properties of SIS under a range of settings. Section \ref{sec:section 4} introduces the GLSS procedure and presents its sure screening properties. Combining these screening procedures with the adaptive Lasso will discussed in Section \ref{sec:section 5}. Section \ref{sec:section 6} covers simulation results, while section \ref{sec:section 7} discusses an application to forecasting the US inflation rate. Lastly, concluding remarks are in Section \ref{sec:section 8}, and the proofs for all the results follow in the appendix.

\section{Preliminaries}\label{sec:section 2}

We shall assume the error sequence is a strictly stationary, ergodic process with the following form:
\begin{equation}\epsilon_i=g\left(\ldots,e_{i-1},e_i\right)\label{nonlinear}\end{equation} 

\noindent Where $g(\cdot)$ is a real valued measurable function, and $e_i$ are iid random variables. This representation includes a very wide range of stochastic processes such as linear processes, their non-linear transforms, Volterra processes, Markov chain models, non-linear autoregressive models such as threshold auto-regressive (TAR), bilinear, GARCH models, among others (for more details see \cite{Wu2011},\cite{Wu2005}). This representation allows us to use the functional and predictive dependence measures introduced in \cite{Wu2005}. The functional dependence measure for the error process is defined as the following:\begin{equation}\delta_{q}(\epsilon_i)=||\epsilon_i-g\left(\mathcal{F}_i^*\right)||_q = (E|\epsilon_i-g\left(\mathcal{F}_i^*\right)|^q)^{1/q}\end{equation}where $\mathcal{F}_i^*=\left(\ldots,e_{-1},e_0^*,e_1,\ldots,e_i\right)$ with $e_0^*,e_j, j \in \mathbb{Z} $ being iid. Since we are replacing $e_0$ by $e_0^*$, we can think of this as measuring the dependency of $\epsilon_i$ on $e_0$ as we are keeping all other inputs the same. The cumulative functional dependence measure is defined as $\Delta_{m,q}(\bm{\epsilon})=\sum_{i=m}^{\infty}\delta_{q}(\epsilon_i)$. We assume weak dependence of the form:
\begin{equation}\Delta_{0,q}(\bm{\epsilon})=\sum_{i=0}^{\infty}\delta_{q}(\epsilon_i)<\infty\label{qstrongstable}\end{equation}
The predictive dependence measure is related to the functional dependence measure, and is defined as the following:\begin{equation}\theta_{q}(\epsilon_l)=||\textrm{E}\left(\epsilon_l|\mathcal{F}_{0}\right)-\textrm{E}\left(\epsilon_l|\mathcal{F}_{-1}\right)||_q= ||\mathcal{P}_0 \epsilon_{l}||_q  \end{equation} where $\mathcal{F}_i=\left(\ldots,e_{-1},e_0,e_1,\ldots,e_i\right)$ with $e_i, i \in \mathbb{Z}$ being iid. The cumulative predictive dependence measure is defined as $\Theta_{q}(\bm{\epsilon})=\sum_{l=0}^{\infty}\theta_{q}(\epsilon_l)$, and by Theorem 1 in \cite{Wu2005} we obtain $\Theta_{q}(\bm{\epsilon}) \leq \Delta_{0,q}(\bm{\epsilon})$.

Similarly the covariate process is of the form:
\begin{equation}\bm{x_i}^{(n)}=\bm{h}\left(\ldots,\bm{\eta}_{i-1}^{(n)},\bm{\eta}_i^{(n)}\right)\label{nonlinear2}\end{equation} 
Where $\bm{\eta}_i^{(n)} \in \mathcal{R}^{p_n}, i\in \mathbb{Z}$, are iid random vectors, $\bm{h(\cdot)}=(h_1(\cdot)\ldots,h_{p_n}(\cdot))$, $\bm{x_i}^{(n)}=(X_{i1},...,X_{ip_n})$ and $X_{ij}=h_j(...,\bm{\eta}_{i-1}^{(n)},\bm{\eta}_{i}^{(n)})$. The superscript $(n)$ denotes that the dimension of vectors is a function of $n$, however for presentational clarity we suppress the superscript $(n)$ from here on and use $\bm{x}_i$ and $\bm{\eta}_i$ instead. Let $\mathcal{H}_{i}^{*}=(\ldots,\bm{\eta}_{-1},\bm{\eta}^{*}_{0},\bm{\eta}_{1},\ldots,\bm{\eta}_{i})$. As before the functional dependence measure is $\delta_{q}(X_{ij})=||X_{ij}-h_{j}\left(\mathcal{H}_i^*\right)||_q $ and the cumulative dependence measure for the covariate process is defined as:
\begin{equation}\Phi_{m,q}(\bm{x})=\sum_{i=m}^{\infty}\max_{j \leq p_n}\delta_{q}(X_{ij}) < \infty \end{equation}
The representations (\ref{nonlinear}), and (\ref{nonlinear2}), along with the functional and predictive dependence measures have been used in various works including \cite{Wu2009},\cite{Wu2012}, and \cite{WuandWu2016} amongst others. Compared to strong mixing conditions, which are often difficult to verify, the above dependence measures are easier to interpret and compute since they are related to the data generating mechanism of the underlying process \cite{Wu2011}. In many cases using the functional dependence measure also requires less stringent assumptions. For example, consider the case of a linear process, $\epsilon_i=\sum_{j=0}^{\infty} f_j e_{i-j}$, with $e_i$ iid. Sufficient conditions for a linear process to be strong mixing involve: the density function of the innovations ($e_i$) being of bounded variation, restrictive assumptions on the decay rate of the coefficients ($f_j$), and invertibility of the process (see Theorem 14.9 in \cite{Davidson94} for details). Additional conditions are needed to ensure strong mixing if the innovations for the linear process are dependent \cite{Doukhan94}. 

As a result many simple processes can be shown to be non-strong mixing. A prominent example involves an AR(1) model with iid Bernoulli (1/2) innovations: $\epsilon_i=\rho\epsilon_{i-1}+e_i$ is non-strong mixing if $\rho \in (0,1/2]$ \cite{Andrews1984}. These cases can be handled quite easily in our framework, since we are not placing distributional assumptions on the innovations, $e_i$, such as the existence of a density. For linear processes with iid innovations, representation (\ref{nonlinear}) clearly holds and (\ref{qstrongstable}) is satisfied if $\sum_{j=0}^{\infty}|f_j|<\infty$. For dependent innovations, suppose we have: $e_i=h(\ldots,a_{i-1},a_i)$, where $h(\cdot)$ is a real valued measurable function and $a_i,i \in \mathbb{Z}$, are iid. Then $\epsilon_i=\sum_{j=0}^{\infty} f_j e_{i-j}$, has a causal representation, and satisfies (\ref{qstrongstable}) if: $\sum_{i=0}^{\infty}\delta_{q}(e_i)<\infty$, and $\sum_{j=0}^{\infty}|f_j|<\infty$ (see \cite{WuMin2005}). 

\section{SIS with Dependent Observations}\label{sec:section 3}
Sure Independence Screening, as introduced by Fan and Lv \cite{FanLv2008}, is a method of variable screening based on ranking the magnitudes of the $p_n$ marginal regression estimates. Under appropriate conditions, this simple procedure is shown to possess the sure screening property. The method is as follows, let: 
\begin{equation} \bm{\hat{\rho}}=(\hat{\rho}_1,\ldots,\hat{\rho}_{p_n})\textrm{, where } \hat{\rho}_j=(\sum_{t=1}^{n}X_{tj}^2)^{-1}(\sum_{t=1}^{n}X_{tj}Y_t) \end{equation}
Therefore, $\hat{\rho}_j$ is the OLS estimate of the linear projection of $Y_t$ onto $X_{tj}$. Now let \begin{equation}\mathcal{M}_*=\left\{1\leq i\leq p_n:\beta_i \neq 0\right\}\end{equation} and let $|\mathcal{M}_*| =s_n<< n$ be the size of the true sparse model. We then sort the elements of $\bm{\hat{\rho}}$ by their magnitudes. For any given $\gamma_n$, define a sub-model 
\begin{equation} \hat{\mathcal{M}}_{\gamma_n}=\left\{1\leq i\leq p_n:|\hat{\rho}_i| \geq \gamma_n \right\}\end{equation} 
\noindent and let $|\hat{\mathcal{M}}_{\gamma_n}| =d_n$ be the size of the selected model. The sure screening property states that for an appropriate choice of $\gamma_n$, we have $P\left(\mathcal{M}_*\subset \hat{\mathcal{M}}_{\gamma_{n}}\right)\rightarrow 1$. 
\newline\newline
Throughout this paper let: $Y_t=\sum_{i=1}^{p_n}X_{ti}\beta_i + \epsilon_t$, $\bm{x_t}=(X_{t1},...,X_{tp_n})$, $\Sigma=cov(\bm{x_t})$, and $\bm{X}_k$ be $k^{th}$ column of $\bm{X}$. In addition, we assume $Var(Y_t), Var(X_{tj})= O(1)$, $\forall j\leq p_n$. Note that $\bm{x}_t$ can contain lagged values of $Y_t$. Additionally, let $\rho_j=(E(X_{tj}^2))^{-1}E(X_{tj}Y_t)$, and $\mathcal{M}_{\gamma_n}=\left\{1\leq i\leq p:|\rho_i| \geq \gamma_n \right\}$. For a vector $\bm{a}=(a_1,...,a_n)$, $sgn(\bm{a})$ denotes its sign vector, with the convention that $sgn(0)=0$, and $|\bm{a}|_p^{p}=\sum_{i=1}^{n} |a_i|^p$. For a square matrix $\bm{A}$, let $\lambda_{\min}(\bm{A})$ and $\lambda_{\max}(\bm{A})$, denote the minimum eigenvalue, and maximum eigenvalue respectively. For any matrix $\bm{A}$, let $||\bm{A}||_{\infty}$, and $||\bm{A}||_2$ denote the maximum absolute row sum of $\bm{A}$, and the spectral norm of $\bm{A}$ respectively. Lastly we will use $C,c$ to denote generic positive constants which can change between instances.

\subsection{SIS with dependent, heavy tailed covariates and errors} To establish sure screening properties, we need the following conditions:
\newline\newline
\noindent \textbf{Condition A}: $|\rho_k| \geq c_1 n^{-\kappa}$ for $k \in M_* \textrm{, }\kappa < 1/2$
\newline\newline
\noindent \textbf{Condition B}: $E(\epsilon_t), E(X_{tj}), E(X_{tj}\epsilon_{t}) =0 \textrm{ }\forall j,t$. 
\newline\newline
\noindent \textbf{Condition C}: Assume the error and the covariate processes have representations (\ref{nonlinear}), and (\ref{nonlinear2}) respectively. Additionally, we assume the following decay rates $\Phi_{m,r}(\bm{x}) = O(m^{-\alpha_x}), \Delta_{m,q}(\bm{\epsilon}) =O(m^{-\alpha_{\epsilon}})$, for some $\alpha_x, \alpha_{\epsilon} > 0$, $q > 2, r > 4 $ and $\tau=\frac{qr}{q+r} > 2$. 
\newline\newline
\indent Condition \textbf{A} is standard in screening procedures, and it assumes the marginal signals of the active predictors cannot be too small. Condition \textbf{B} assumes the covariates and the errors are contemporaneously uncorrelated. This is significantly weaker than independence between the error sequence and the covariates usually assumed. Condition \textbf{C} presents the structure, dependence and moment conditions on the covariate and error processes. Notice that higher values of $\alpha_x,\alpha_{\epsilon}$ indicate weaker temporal dependence. 

Examples of error and covariate processes which satisfy Condition C are: If $\epsilon_i$ is a linear process, $\epsilon_i=\sum_{j=0}^{\infty} f_j e_{i-j}$ with $e_i$ iid and $\sum_{j=0}^{\infty}|f_j|<\infty$ then $\delta_{q}(\epsilon_i)=|f_i| ||e_{0}-e_{0}^{*}||_{q}$. If $f_i=O(i^{-\beta})$ for $\beta>1$ we have $\Delta_{m,q}=O(m^{-\beta+1})$ and $\alpha_{\epsilon}=\beta-1$. 
We have a geometric decay rate in the cumulative functional dependence measure, if $\epsilon_{i}$ satisfies the geometric moment contraction (GMC) condition, see \cite{Wu2007}. Conditions needed for a process to satisfy the GMC condition are given in Theorem 5.1 of \cite{Wu2007}. Examples of processes satisfying the GMC condition include stationary, causal finite order ARMA, GARCH, ARMA-GARCH, bilinear, and threshold autoregressive processes, amongst others (see \cite{Wu2011} for details).

For the covariate process, if we assume $\bm{x}_i$ is a vector linear process: $\bm{x}_i=\sum_{l=0}^{\infty} A_l\bm{\eta}_{i-l}$. Where $A_{l}$ are $p_n \times p_n$ coefficient matrices and $\bm{\eta}_i=(\eta_{i1},\ldots,\eta_{ip_n})$ are iid random vectors with $cov(\bm{\eta_i})=\Sigma_{\eta}$. For simplicity, assume $\eta_{i,j} (j=1,\ldots,p_n)$ are identically distributed, then 
\begin{equation} \delta_{q}(X_{ij})=||A_{i,j}\bm{\eta}_0-A_{i,j}\bm{\eta}_0^{*}||_q \leq 2|A_{i,j}|||\eta_{0,1}||_q \end{equation}
where $A_{i,j}$ is the $j^{th}$ column of $A_i$. If $||A_i||_{\infty}=O(i^{-\beta})$ for $\beta>1$, then $\Phi_{m,q}=O(m^{-\beta+1})$.

In particular for stable VAR(1) processes, $\bm{x}_t=B_1\bm{x}_{t-1} + \bm{\eta}_t$, $\Phi_{m,q}(\bm{x})=O(||\tilde{B}_1||_2^{m})$ \cite{chen2013}. For stable VAR($k$) processes, $\bm{x}_t=\sum_{i=1}^{k} B_i\bm{x}_{t-i} + \bm{\eta}_t$, we can rewrite this as a VAR(1) process, $\tilde{\bm{x}}_t=\tilde{B}_1\tilde{\bm{x}}_{t-1} + \tilde{\bm{\eta}}_t$, with:

\begin{align}
\tilde{\bm{x}}_t &= \begin{bmatrix}
           \bm{x}_{t} \\
           \bm{x}_{t-1} \\
           \vdots \\
           \bm{x}_{t-k+1}
         \end{bmatrix}_{kp \times 1}\textrm{   }         
\tilde{B}_1
  = \left(
\begin{matrix}
      B_1 & \cdots & B_{k-1} & B_k \\
      I_{p_n} & \cdots & \bm{0} & \bm{0}  \\
      \vdots & \ddots & \vdots & \vdots \\
      \bm{0} & \cdots & I_{p_n} & \bm{0}\\
   \end{matrix}\right)_{kp \times kp}
\tilde{\bm{\eta}}_t &= \begin{bmatrix}
           \bm{\eta}_{t} \\
           \bm{0} \\
           \vdots \\
           \bm{0}
         \end{bmatrix}_{kp \times 1}
\end{align}

And by section 11.3.2 in \cite{Lutkepohl}, the process $\tilde{\bm{x}}_t$ is stable if and only if $\bm{x}_t$ is stable. Therefore if $\tilde{B}_1$ is diagonalizable, we have $O(a^{m})$, where $a$ represents the largest eigenvalue in magnitude of $\tilde{B}_1$. And by the stability of $\bm{x}_t$, $a \in (0,1)$. Additional examples of error and covariate processes which satisfy Condition C are given in \cite{Wu2009} and \cite{WuandWu2016} respectively.

Define $\alpha = min(\alpha_x,\alpha_{\epsilon})$ and let $\omega =1$ if $\alpha_x > 1/2 - 2/r$, otherwise $\omega= r/4 - \alpha_x r/2$. Let $\iota = 1$ if $\alpha > 1/2-1/\tau$, otherwise $\iota = \tau/2-\tau\alpha$. Additionally, let $K_{\epsilon,q}=\sup_{m \geq 0} (m+1)^{\alpha_{\epsilon}}\Delta_{m,q}(\bm{\epsilon})$ and $K_{x,r}=\max_{j \leq p_n} \sup_{m \geq 0} (m+1)^{\alpha_x} \sum_{i=m}^{\infty}\delta_{r}(X_{ij})$. Given Condition C, it follows that $K_{\epsilon,q}, K_{x,r} < \infty$. For ease of presentation we let:
\begin{align}
\vartheta_n=\frac{s_nn^{\omega} K_{x,r}^{r}}{(n/s_n)^{r/2-r\kappa/2}}+\frac{n^{\iota} K_{x,r}^{\tau}K_{\epsilon,q}^{\tau}}{n^{\tau-\tau\kappa}} + \exp\left(-\frac{n^{1-2\kappa}}{s_n^2K_{x,r}^4}\right)+ \exp\left(-\frac{n^{1-2\kappa}}{K_{x,r}^2K_{\epsilon,q}^2}\right)
\end{align}

The following theorem gives the sure screening properties, and provides a bound on the size of the selected model:
\begin{thm}{} Suppose Conditions A,B,C hold.  
\begin{enumerate}[(i)]
\item For any $c_2 > 0$, we have:
\begin{align*} P\left(\max_{j \leq p_n}|\hat{\rho}_j -\rho_j|>c_2 n^{-\kappa}\right)
&\leq O(p_n\vartheta_n)
\end{align*}
\item For $\gamma_n = c_3 n^{-\kappa}$ with $c_3 \leq c_1/2$, we have:
\begin{align*} P\left(\mathcal{M}_*\subset \hat{\mathcal{M}}_{\gamma_n} \right)&\geq 1 - O(s_n\vartheta_n)
\end{align*}
\item For $\gamma_n = c_3 n^{-\kappa}$ with $c_3 \leq c_1/2$, we have:
\begin{align*} P\left(|\hat{\mathcal{M}}_{\gamma_n}| \leq O(n^{2\kappa}\lambda_{max}(\Sigma))\right)
&\geq 1- O(p_n\vartheta_n)
\end{align*}
\end{enumerate}
\label{theorem2} 
\end{thm}

In Theorem \ref{theorem2} we have two types of bounds, for large $n$ the polynomial terms dominate, whereas for small values of $n$ the exponential terms dominate. The covariate dimension ($p_n$) can be as large as $o(min(\frac{s_n(n/s_n)^{r/2-r\kappa/2}}{n^{\omega}},\frac{n^{\tau-\tau\kappa}}{n^{\iota}}))$. The range of $p_n$ depends on the dependence in both the covariate and the error processes, the strength of the signal ($\kappa$), the moment condition, and the sparsity level ($s_n$). If we assume $s_n=O(1), r=q$, and $\alpha \geq 1/2 - 2/r$ then $p_n=o(n^{r/2-r\kappa/2-1})$. For the case of iid errors and covariates, we would replace $K_{x,r}, K_{\epsilon,q}$ in Theorem \ref{theorem2} with $\max_{j\leq p_n} ||X_{ij}||_{r/2}$ and $||\epsilon_i||_q$ respectively. Therefore for the case of weaker dependence in the covariate and error processes (i.e. $\alpha_x >1/2 - 2/r$ and $\alpha >1/2 - 1/\epsilon$), our range for $p_n$ is reduced only by a constant factor. However, our range for $p_n$ is significantly reduced in the case of stronger dependence in the error or covariate processes (i.e. either $\alpha_x <1/2 - 2/r$ or $\alpha_{\epsilon} < 1/2 - 2/q$). For instance if $\alpha_x=\alpha_{\epsilon}$ and $q=r$, our range for $p_n$ is reduced by a factor of $n^{r/4-\alpha r/2}$ in the case of stronger dependence.

In the iid setting, to achieve sure screening in the ultrahigh dimensional case, \cite{FanLv2008} assumed the covariates and errors are jointly normally distributed. Future works applicable to the linear model, such as \cite{FanSong2010},\cite{Feng2011} among others, relaxed this Gaussian assumption, but generally assumed the tails of the covariates and errors are exponentially light. Compared to the existing results for iid observations, our moment conditions preclude us from dealing with the ultrahigh dimensional case. However, our setting is far more general in that it allows for dependent and heavy tailed covariates and errors. In addition, we allow for the covariates and error processes to be dependent on each other, with the mild restriction that $ E(X_{tj}\epsilon_{t}) =0$, $\forall j \leq p_n$.

\subsection{Ultrahigh Dimensionality under dependence}

It is possible to achieve the sure screening property in the ultrahigh dimensional setting with dependent errors and covariates. However, we need to make stronger assumptions on the moments of both the error and covariate processes. Until now we have assumed the existence of a finite $qth$ moment, which restricted the range of $p$ to a power of $n$. If the error and covariate processes are assumed to follow a stronger moment condition, such as $\Delta_{0,q}(\bm{\epsilon})<\infty$ and $\Phi_{0,q}(\bm{x}) < \infty$ for arbitrary $q>0$, we can achieve a much larger range of $p_n$ which will cover the ultrahigh dimensional case discussed in \cite{FanLv2008}. More formally, we have:
\newline\newline
\noindent \textbf{Condition D}: Assume the error and the covariate processes have representations (\ref{nonlinear}), and (\ref{nonlinear2}) respectively. Additionally assume $\upsilon_x=\sup_{q \geq 2} q^{-\tilde{\alpha}_x}\Phi_{0,q}(\bm{x}) < \infty$ and $\upsilon_{\epsilon}=\sup_{q \geq 2} q^{-\tilde{\alpha}_{\epsilon}} \Delta_{0,q}(\bm{\epsilon})< \infty$, for some $\tilde{\alpha}_x,\tilde{\alpha}_{\epsilon} \geq 0.$
\newline\newline
By Theorem 3 in \cite{WuandWu2016}, Condition D implies the tails of the covariate and error processes are exponentially light. There are a wide range of processes which satisfy the above condition. For example, if $\epsilon_i$ is a linear process: $\epsilon_i=\sum_{j=0}^{\infty} f_j e_{i-j}$ with $e_i$ iid and $\sum_{l=0}^{\infty}|f_l|<\infty$ then $\Delta_{0,q}(\epsilon_l)= ||e_{0}-e_{0}^{*}||_{q}\sum_{l=0}^{\infty}|f_l|$. If we assume $e_0$ is sub-Gaussian, then $\tilde{\alpha}_{\epsilon}=1/2$, since $||e_{0}||_{q} = O(\sqrt{q})$. Similarly if $e_i$ is sub-exponential we have $\tilde{\alpha}_{\epsilon}=1$. More generally, for $\epsilon_i=\sum_{j=0}^{\infty}f_je_{i-j}^{p}$, if $e_i$ is sub-exponential, we have $\tilde{\alpha}_{\epsilon} = p$. Similar results hold for vector linear processes discussed previously. 

Condition D is primarily a restriction on the rate at which $||\epsilon_i||_q, \max_{j\leq p_n}||X_{ij}||_q$ increase as $q \rightarrow \infty$. We remark that, for any fixed $q$, we are not placing additional assumptions on the temporal decay rate of the covariate and error processes besides requiring $\Delta_{0,q}(\bm{\epsilon})$,$\Phi_{0,q}(\bm{x}) < \infty$. In comparison, in the ultrahigh dimensional setting, \cite{Chenlu2017} requires geometrically decaying strong mixing coefficients, in addition to requiring sub-exponential tails for the response. As an example, if we assume $\epsilon_i=\sum_{j=0}^{\infty} f_j e_{i-j}$, geometrically decaying strong mixing coefficients would require the coefficients, $f_j$, to decay geometrically. Whereas in Condition D, the only restrictions we place on the coefficients, $f_j$, is absolute summability.

\begin{thm} Suppose Conditions A,B,D hold. Define $\tilde{\alpha}'=\frac{2}{1+2\tilde{\alpha}_x+2\tilde{\alpha}_{\epsilon}}$, and $\tilde{\alpha} = \frac{2}{1+4\tilde{\alpha}_x}$.
\begin{enumerate}[(i)]
\item For any $c_2 > 0 $ we have:
\begin{align*} P\left(\max_{j \leq p_n}|\hat{\rho}_j -\rho_j|>c_2 n^{-\kappa}\right)
\leq & O\left(s_np_n\exp\left(-\frac{n^{1/2-\kappa}}{\upsilon_x^{2}s_n}\right)^{\tilde{\alpha}}\right) \nonumber \\
& + O\left(p_n\exp\left(-\frac{n^{1/2-\kappa}}{\upsilon_x\upsilon_{\epsilon}}\right)^{\tilde{\alpha}'}\right)
\end{align*}
\item For $\gamma_n = c_3 n^{-\kappa}$ with $c_3 \leq c_1/2$, we have:
\begin{align*} P\left(\mathcal{M}_*\subset \hat{\mathcal{M}}_{\gamma_n} \right)\geq 1 &- O\left(s_n^2 \exp\left(-\frac{n^{1/2-\kappa}}{\upsilon_x^{2}s_n}\right)^{\tilde{\alpha}}\right)\\
&- O\left(s_n\exp\left(-\frac{n^{1/2-\kappa}}{\upsilon_x\upsilon_{\epsilon}}\right)^{\tilde{\alpha}'}\right)
\end{align*}
\item For $\gamma_n = c_3 n^{-\kappa}$ with $c_3 \leq c_1/2$, we have:
\begin{align*} P\left(|\hat{\mathcal{M}}_{\gamma_n}| \leq O(n^{2\kappa}\lambda_{max}(\Sigma))\right)
\geq  1&- O\left(s_np_n\exp\left(-\frac{n^{1/2-\kappa}}{\upsilon_x^{2}s_n}\right)^{\tilde{\alpha}}\right) \nonumber \\
& - O\left(p_n\exp\left(-\frac{n^{1/2-\kappa}}{\upsilon_x\upsilon_{\epsilon}}\right)^{\tilde{\alpha}'}\right)
\end{align*}

\end{enumerate}
\label{theorem4} 
\end{thm}

From Theorem \ref{theorem4}, we infer the covariate dimension ($p_n$) can be as large as \\ $o(\min[\exp\left(\frac{Cn^{1/2-\kappa}}{s_n}\right)^{\tilde{\alpha}}/s_n, \exp(Cn^{1/2-\kappa})^{\tilde{\alpha}'}])$. As in Theorem \ref{theorem2}, the range of $p_n$ depends on the dependence in both the covariate and the error processes, the strength of the signal ($\kappa$), the moment condition, and the sparsity level ($s_n$). For the case of iid covariates and errors, we would replace $\upsilon_x$ and $\upsilon_{\epsilon}$ with $\mu_{r/2}=\max_{j\leq p_n} ||X_{ij}||_{r/2}$ and $||\epsilon_i||_q$ respectively. In contrast to Theorem \ref{theorem2}, temporal dependence affects our range of $p_n$ only by a constant factor. 

If we assume $s_n=O(1)$, and both the covariate and error processes are sub-Gaussian we obtain $p_n=o(\exp(n^{\frac{1-2\kappa}{3}}))$, while for sub-exponential distributions we obtain $p_n=o(\exp(n^{\frac{1-2\kappa}{5}}))$. In contrast, Fan and Lv \cite{FanLv2008}, assuming independent observations, allow for a larger range $p_n=o(\exp(n^{1-2\kappa})$. However, their work relied critically on the Gaussian assumption. Fan and Song \cite{FanSong2010}, relax the Gaussian assumption by allowing for sub-exponential covariates and errors, and our rates are similar to theirs up to a constant factor. Additionally, in our work we relax the sub-exponential assumption, provided the tails of the covariates and errors are exponentially light.

\section{Generalized Least Squares Screening (GLSS)}\label{sec:section 4}

Consider the marginal model:
\begin{equation} Y_t=X_{tk}\rho_k + \epsilon_{t,k} \label{marg}\end{equation}
where $\rho_k$ is the linear projection of $y_t$ onto $X_{tk}$. In SIS, we rank the magnitudes of the OLS estimates of this projection. In a time series setting, if we are considering the marginal model (\ref{marg}) it is likely the case that the marginal errors ($\epsilon_{t,k}$) will be serially correlated. This holds even if we assume that the errors ($\epsilon_t$) in the full model (\ref{eq1}) are serially uncorrelated. A procedure which accounts for this serial correlation, such as Generalized Least Squares (GLS), will provide a more efficient estimate of $\rho_k$. 
\begin{figure}
\includegraphics[scale=.5]{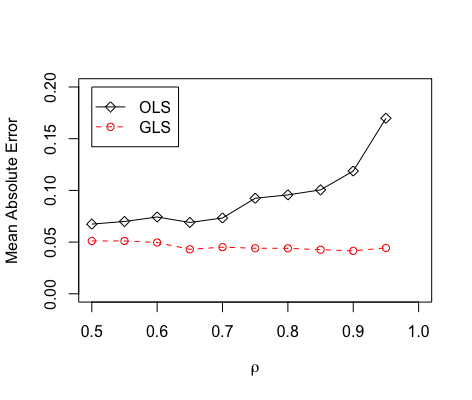} 
\caption{GLS vs OLS error comparison for values of $\rho$ between .5 and .95 incrementing by .05. Absolute error averaged over 200 replications.}
\label{figure1}
\end{figure}

We first motivate our method by considering a simple univariate model. Assume $Y_t=\beta X_{t}+\epsilon_{t}$ and the errors follow an AR(1) process, $\epsilon_t=\rho\epsilon_{t-1}+\theta_t$, where $\theta_t$, and $X_t$ are iid standard Gaussian. We set $\beta=.5$, $n=200$, and estimate the model using both OLS and GLS for values of $\rho$ ranging from .5 to .95. The mean absolute errors for both procedures is plotted in figure \ref{figure1}. We observe that the performance of OLS steadily deteriorates for increasing values of $\rho$, while the performance of GLS stays constant. This suggests that a screening procedure based on GLS estimates will be most useful in situations where we have weak signals and high levels of serial correlation.
 
 The infeasible GLS estimate for $\rho_k$ is:
\begin{equation} \tilde{\beta}_k^{M} =(\bm{X}_k^{T}\Sigma_k^{-1}\bm{X}_k)^{-1} \bm{X}_k^T\Sigma_k^{-1}\bm{y} \end{equation} 
Where $\bm{X}_k$ is the $k^{th}$ column of $\bm{X}$, and $\Sigma_k=(\gamma_{i-j,k})_{1 \leq i,j \leq n}$ is the auto-covariance matrix of $\bm{\epsilon}^k=(\epsilon_{t,k}, t=1,...,n)$. 
Given that $\Sigma_k$ needs to be estimated to form our GLS estimates, we use the banded autocovariance matrix estimator introduced in \cite{Wu2009}, which is defined as:
\begin{equation} \hat{\Sigma}_{k,l_n}= \left(\hat{\gamma}_{i-j,k} \mathbbm{1}_{|i-j| \leq l_n}\right)_{1\leq i,j \leq n}\end{equation}
Where $l_n$ is our band length, $\hat{\gamma}_{r,k}=\frac{1}{n}\sum_{t=1}^{n-|r|} \hat{\epsilon}_{t,k}\hat{\epsilon}_{t+|r|,k}$, with $\hat{\epsilon}_{t,k}=y_t-X_{tk}\hat{\rho}_k$, and $\hat{\rho}_k$ is the OLS estimate of $\rho_k$. Our GLS estimator is now:
\begin{equation} \hat{\beta}_k^{M} =(\bm{X}_k^{T}\hat{\Sigma}_{k,l_n}^{-1}\bm{X}_k)^{-1} \bm{X}_k^T\hat{\Sigma}_{k,l_n}^{-1}\bm{y} \label{est}\end{equation} 

When $E(\bm{\epsilon}^k|\bm{X}_k)=0$, by the Gauss-Markov theorem it is clear that $\tilde{\beta}_k^M$ is efficient relative to the OLS estimator. \cite{amemiya1973} showed that under non-stochastic regressors and appropriate conditions on the error process, a two stage sieve type GLS estimator has the same limiting distribution as the infeasible GLS estimator $\tilde{\beta}_k^M$. In the appendix, we provide the appropriate conditions under which our GLS estimator, $\hat{\beta}_k^M$, and the infeasible GLS estimate, $\tilde{\beta}_k^M$, have the same asymptotic distribution.

For positive definite $\Sigma_k$, the banded estimate for $\Sigma_k$ is not guaranteed to be positive definite, however it is asymptotically positive definite (see Lemma \ref{lem2}). For small samples, we can preserve positive definiteness by using the tapered estimate: $\hat{\Sigma}_{k} * R_{l_n}$, where $R_{l_n}$ is a positive definite kernel matrix, and $*$ denotes coordinate-wise multiplication. For example, we can choose $R_{l_n}=(\max(1-\frac{|i-j|}{l_n},0))_{1\leq i,j \leq n}$. 
We need the following conditions for the sure screening property to hold:
\newline\newline
\noindent \textbf{Condition E}: Assume the marginal error process, $\epsilon_{t,k}$, is a stationary AR($L_k$) process, $\epsilon_{t,k}=\sum_{i=1}^{L_k} \alpha_i \epsilon_{t-i,k}+e_t$. Where $L_k <K<\infty$, $\forall k \leq p_n$. 
\newline\newline
\noindent \textbf{Condition F}: For $k \in M_* \textrm{, } \kappa < 1/2$: $\beta_{k}^M=E(y_t-\sum_{i=1}^{L_k} \alpha_iy_{t-i})(X_{t,k}-\sum_{i=1}^{L_k} \alpha_iX_{t-i,k})/(E(X_{t,k}-\sum_{i=1}^{L_k} \alpha_iX_{t-i,k})^2) \geq c_6n^{-\kappa}$. 
\newline\newline
\noindent \textbf{Condition G}: Assume $E(X_{tk}),E(\epsilon_{t}),E(\bm{X}_k^T\Sigma_k^{-1}\bm{\epsilon})=0$  
\newline\newline
\noindent \textbf{Condition H}: Assume $\epsilon_{t,k}$, $\epsilon_t$ are of the form (\ref{nonlinear}), and the covariate process is of the form (\ref{nonlinear2}). Additionally we assume the following decay rates $\Delta_{m,q}(\bm{\epsilon})=O(m^{-\alpha_{\epsilon}})$, $\Phi_{m,r}(\bm{x}) = O(m^{-\alpha_x}), \chi_{m,q'} =\sum_{i=m}^{\infty}\max_{k \leq p_n}\delta_{q}(\epsilon_{i,k}) = O(m^{-\alpha})$, for some $\alpha_x,\alpha_{\epsilon} > 0$, $\alpha = min(\alpha_x,\alpha_{\epsilon})$, and $q'=\min(q,r) \geq 4$.
\newline\newline
\noindent \textbf{Condition I}: Assume $\epsilon_{t,k}$, $\epsilon_t$ are of the form (\ref{nonlinear}), and the covariate process is of the form (\ref{nonlinear2}). Additionally assume $\upsilon_x=\sup_{q \geq 4} q^{-\tilde{\alpha}_x}\Phi_{0,q}(\bm{x}) < \infty$ \\ 
,$\upsilon_{\epsilon}=\sup_{q \geq 4} q^{-\tilde{\alpha}_{\epsilon}} \Delta_{0,q}(\bm{\epsilon})< \infty, \phi=\sup_{q \geq 4} q^{-\varphi}\chi_{0,q} < \infty$ for some $\tilde{\alpha}_x,\tilde{\alpha}_{\epsilon} \geq 0$, and $\varphi=\max(\tilde{\alpha}_{\epsilon},\tilde{\alpha}_{x})$.
\newline\newline
%\noindent\textbf{Condition M:} Assume $|\gamma_{t,k}|=O(t^{-b})$ where $b > 1$ for $k \leq p_n$. Additionally, $l_n \rightarrow \infty$, $l_n^{-b+1}=o(n^{-\kappa})$, and $l_n=o(n^{1/2-\kappa})$
\indent In Condition E, we can let the band length $K$ diverge to infinity at a slow rate, e.g $O(\log(n))$, for simplicity we set $K$ to be a constant. Assuming a finite order AR model for the marginal error process is reasonable in most practical situations, since any stationary process with a continuous spectral density function can be approximated arbitrarily closely by a finite order linear AR process (see corollary 4.4.2 in \cite{Davis91}). For further details on linear AR approximations to stationary processes, see \cite{amemiya1973} and \cite{buhlman95}. We remark that compared to previous works \cite{amemiya1973,Koreisha2001}, knowledge about the structure of the marginal errors is not necessary in estimating $\beta_k^{M}$, since we use a non-parametric estimate of $\Sigma_k$. Therefore Condition E is assumed strictly for technical reasons.

For Condition F, from (\ref{marg}), we have $\beta_k^{M}=\rho_k$, iff $E(\epsilon_{t,k}-\sum_{i=1}^{L_k} \alpha_i\epsilon_{t-i,k})(X_{t,k}-\sum_{i=1}^{L_k} \alpha_iX_{t-i,k})=0$. When $\beta_k^{M} \neq \rho_k$, recall that: 
\begin{equation} \beta_{k}^M=E(y_t-\sum_{i=1}^{L_k} \alpha_iy_{t-i})(X_{t,k}-\sum_{i=1}^{L_k} \alpha_iX_{t-i,k})/(E(X_{t,k}-\sum_{i=1}^{L_k} \alpha_iX_{t-i,k})^2) \end{equation}
If we assume the cross covariance, $\gamma_{\bm{X}_k,Y}(h)$, is proportional to $E(X_{tk}Y_t)$, i.e. $\gamma_{\bm{X}_k,Y}(h) \propto E(X_{tk}Y_t)$, for $h\in \{-L_k,\ldots,-1,1,\ldots,L_k\} $, then $\beta_k^M \propto \rho_k$ whenever $|\beta_k^M|>0$. And for $|\rho_k|>0$, it is likely the case that $\beta_k^M \propto \rho_k$ if we assume $\gamma_{\bm{X}_k,Y}(h) \propto E(X_{tk}Y_t),\textrm{ for } h\in \{-L_k,\ldots,-1,1,\ldots,L_k\}$. When $\beta_k^{M} \neq \rho_k$, we believe the advantage in using GLSS is due to the GLS estimator being robust to serial correlation in the marginal error process (see the appendix for details).

For Condition H, since $\epsilon_{t,k}=Y_t-X_{tk}\rho_k $, we have $\epsilon_{t,k}=r_k(\ldots,\bm{\theta}_{t-1},\bm{\theta}_t)$, where $r_k(\cdot)$ is a measurable function and $\bm{\theta}_t=(\bm{\eta}_t,e_t)$. If we assume $e_t$, and $\bm{\eta}_i$ are independent for $i \neq t$, then $\bm{\theta}_i$ are iid. We then have:
\begin{align*} \delta_{q'}(\epsilon_{t,k}) = &||\sum_{i \in M_{*}}X_{ti}\beta_i+\epsilon_t-X_{tk}\rho_k-(\sum_{i \in M_{*}}X_{ti}^*\beta_i+\epsilon_t^*-X_{tk}^*\rho_k)||_{q'} \\
 & \leq \sum_{i \in M_{*}} |\beta_i|\delta_{q'}(X_{ti})+\delta_{q'}(\epsilon_t)+|\rho_k|\delta_{q'}(X_{tk}) \end{align*}
Therefore, $\chi_{m,q'}=O(m^{-\alpha})$, if we assume $\sum_{i \in M_{*}}|\beta_i| =O(1)$, $\Delta_{m,q}(\bm{\epsilon})=O(m^{-\alpha_{\epsilon}})$, and $\Phi_{m,r}(\bm{x}) = O(m^{-\alpha_x})$.

For GLSS; define $\hat{\mathcal{M}}_{\gamma_n}=\left\{1\leq i\leq p_n:|\hat{\beta}_k^{M}| \geq \gamma_n \right\}$, $\alpha = min(\alpha_x,\alpha_{\epsilon})$,\\$\tau=\frac{qr}{q+r}, \tau'=\frac{qq'}{q+q'}=\min(q/2,\tau)$. Let $\iota = 1$ if $\alpha > 1/2-1/\tau'$, otherwise $\iota = \tau'/2-\tau'\alpha$. Let $\zeta =1$, if $\alpha> 1/2 - 2/q'$, otherwise $\zeta= q'/4 - \alpha q'/2$ and let $\omega =1$, if $\alpha_x > 1/2 - 2/r$, otherwise $\omega= r/4 - \alpha_x r/2$. Additionally, let $K_{x,r}=\max_{j \leq p_n} \sup_{m \geq 0} (m+1)^{\alpha_x} \sum_{i=m}^{\infty}\delta_{r}(X_{ij})$, $\tilde{K}_{\epsilon,q'}=\max_{k \leq p_n} \sup_{m \geq 0} (m+1)^{\alpha} \sum_{i=m}^{\infty}\delta_{q'}(\epsilon_{i,k})$. Given Condition H, it follows that $K_{x,r},\tilde{K}_{\epsilon,q'}< \infty$. For the case of exponentially light tails, we define $\tilde{\varphi}'=\frac{2}{1+2\tilde{\alpha}_x+2\varphi}, \tilde{\varphi} = \frac{2}{1+4\varphi}$, and $\tilde{\alpha} = \frac{2}{1+4\tilde{\alpha}_x}$. Lastly, for ease of presentation let:
\begin{align}
&a_n= l_n\left[\frac{n^{\iota}l_n^{\tau'}K_{x,r}^{\tau'}\tilde{K}_{\epsilon,q'}^{\tau'}}{n^{\tau'-\tau'\kappa}}+\frac{n^{\zeta}l_n^{q'/2}\tilde{K}_{\epsilon,q'}^{q'}}{n^{q'/2-q'\kappa/2}}+\frac{n^{\omega}l_n^{r/2}K_{x,r}^{r}}{n^{r/2}}\right]\\
\nonumber\\
&b_n=l_n\left[\exp\left(-\frac{n^{1/2}}{l_n\upsilon_x^{2}}\right)^{\tilde{\alpha}}+\exp\left(-\frac{n^{1/2-\kappa}}{l_n\upsilon_x\phi}\right)^{\tilde{\varphi}'}+\exp\left(-\frac{n^{1/2-\kappa}}{l_n\phi^2}\right)^{\tilde{\varphi}}\right]
\end{align}
\indent We first present the following lemma, which provides deviation bounds on $||\hat{\Sigma}_{k,l_n}-\Sigma_k||_2$. This lemma, which is of independent interest, will allow us to obtain deviation bounds on our GLSS estimates.

\begin{lem} Assume the band length, $l_n=c\log(n)$ for sufficiently large $c>0$. 
\begin{enumerate}[(i)]
\item Assume Condition H holds. For $\kappa \in [0,1/2)$ we have the following:
\begin{align*}
P(||\hat{\Sigma}_{k,l_n}-\Sigma_{k}||_2 > cn^{-\kappa}) \leq O(a_n) \end{align*}
\item Assume Condition I holds. For $\kappa \in [0,1/2)$ we have the following:
\begin{align*}
P(||\hat{\Sigma}_{k,l_n}-\Sigma_{k}||_2 > cn^{-\kappa})  \leq O(b_n)
\end{align*}
\end{enumerate}

\label{lem2}
\end{lem}

\noindent The following theorem gives the sure screening properties of GLSS:
\begin{thm} Assume the band length, $l_n=c\log(n)$ for sufficiently large $c>0$.
\begin{enumerate}[(i)]
\item Assume Conditions E,F,G,H hold, for any $c_2 > 0$ we have:
%\begin{equation} P(max_{j \leq p_n}|\rho_j -E(\rho_j)|>c_2 n^{-\kappa}) \end{equation}
\begin{align*} P\left(\max_{j \leq p_n}|\hat{\beta}_k^{M} -\beta_k^{M}|>c_2 n^{-\kappa}\right)
&\leq O(p_na_n)
\end{align*}
\item Assume Conditions E,F,G,H hold, then for $\gamma_n = c_5 n^{-\kappa}$ with $c_5 \leq c_6/2$:
\begin{align*} P\left(\mathcal{M}_*\subset \hat{\mathcal{M}}_{\gamma_n} \right)\geq 1 &- O(s_na_n)
\end{align*}
\item Assume Conditions E,F,G,I hold, for any $c_2>0$ we have:
\begin{align*}P\left(\max_{j \leq p_n}|\hat{\beta}_k^{M} -\beta_k^{M}|>c_2 n^{-\kappa}\right) \leq  O(p_nb_n)
\end{align*}
\item Assume Conditions E,F,G,I hold, then for $\gamma_n = c_5 n^{-\kappa}$ with $c_5 \leq c_6/2$:
\begin{align*} P\left(\mathcal{M}_*\subset \hat{\mathcal{M}}_{\gamma_n} \right)\geq 1 &- O(s_nb_n)
\end{align*}
\end{enumerate}
\label{theorem6} 
\end{thm}

In Lemma \ref{lem2}, the rate of decay also depends on the band length ($l_n$). The band length primarily depends on the decay rate of the autocovariances of the process $\epsilon_{t,k}$. Since we are assuming an exponential decay rate, we can set $l_n=O(\log(n))$. If $\gamma_{i,k}=O(i^{-\beta})$ for $\beta>1$, then we require $l_n^{-\beta+1}=o(n^{-\kappa})$. We omit the exponential terms in the bounds for part part (i) of Lemma \ref{lem2}, and parts (i), and (ii) of Theorem \ref{theorem6} to conserve space and provide a cleaner result. For GLSS, the range for $p_n$ also depends on the band length ($l_n$), in addition to the moment conditions and the strength of dependence in the covariate and error processes. For example, if we assume $r=q$, and $\alpha \geq 1/2 - 2/r$ then $p_n=o(n^{r/2-r\kappa/2-1}/l_n^{r/2+1})$. Compared to SIS, we have a lower range of $p_n$ by a factor of $l_n^{r/2+1}$. We conjecture that this is due to our proof strategy, which relies on using a deviation bound on $||\hat{\Sigma}_{k,l_n}-\Sigma_{k}||_2$, and uses the functional dependence measure, rather than autocorrelation, to quantify dependence. In practice, we believe using GLSS, which corrects for serial correlation, and uses an estimator with lower asymptotic variance will achieve better performance. We illustrate this in more detail in our simulations section, and in the appendix (section \ref{appendixB}).

Similar to SIS, we can control the size of the model selected by GLSS. For the case when $\beta_k^{M}=\rho_k\textrm{ } \forall k$, the bound on the selected model size is the same as in SIS. However, we need to place an additional assumption when $\beta_k^{M}\neq \rho_k$: If the cross covariance, $\gamma_{\bm{X}_k,Y}(h) \propto E(X_{tk}Y_t)$, for $h\in \{-L_k,\ldots,-1,1,\ldots,L_k\}$, we can bound the selected model size by the model size selected by SIS. More formally we have:
\begin{cor} Assume the cross covariance, $\gamma_{\bm{X}_k,Y}(h) \propto E(X_{k,t}Y_t)$, for $h\in \{-L_k,\ldots,-1,1,\ldots,L_k\}$
\begin{enumerate}[(i)]
\item Assume Conditions E,F,G,H hold, then for $\gamma_n = c_5 n^{-\kappa}$ with $c_5 \leq c_6/2$:
\begin{align*} P\left(|\hat{\mathcal{M}}_{\gamma_n}| \leq O(n^{2\kappa}\lambda_{max}(\Sigma))\right) \geq 1 - O(p_na_n) 
\end{align*}
\item Assume Conditions E,F,G,I hold, then for $\gamma_n = c_5 n^{-\kappa}$ with $c_5 \leq c_6/2$:
\begin{align*} P\left(|\hat{\mathcal{M}}_{\gamma_n}| \leq O(n^{2\kappa}\lambda_{max}(\Sigma))\right) \geq 1 - O(p_nb_n) 
\end{align*}
\end{enumerate}
\label{cor1}
\end{cor}

\section{Second Stage Selection with Adaptive Lasso}\label{sec:section 5}

The adaptive Lasso, as introduced by \cite{Zou2006}, is the solution to the following:
\begin{equation} argmin_\beta \textrm{ } ||\bm{y} - \bm{X}\bm{\beta}||^2+\lambda_{n}\sum_{j=1}^{p_n} {w}_{j} |\beta_j|,\textrm{ where } w_{j}=|\hat{\beta}_{I,j}|^{-1}, \end{equation}
and $\hat{\beta}_{I,j}$ is our initial estimate. For sign consistency; when $p_n>>n$, the initial estimates can be the marginal regression coefficients provided the design matrix satisfies the partial orthogonality condition as stated in \cite{Huangetal2008}, or we can use the Lasso as our initial estimator provided the restricted eigenvalue condition holds (see \cite{MM2016}). Both of these conditions can be stringent when $p_n>>n$. This makes the adaptive Lasso a very attractive option as a second stage variable selection method, after using screening to significantly reduce the dimension of the feature space. We have the following estimator:
\begin{equation} \tilde{\bm{\beta}}_{\hat{\mathcal{M}}_{\gamma_n}}=argmin_{\bm{\beta}_{\hat{\mathcal{M}}_{\gamma_n}}} \textrm{ } ||\bm{y} - \bm{X}_{\hat{\mathcal{M}}_{\gamma_n}}\bm{\beta}_{\hat{\mathcal{M}}_{\gamma_n}}||^2+\lambda_{n}\sum_{j=1}^{d_n} {w}_{j} |\beta_j|, w_{j}=|\hat{\beta}_{I,j}|^{-1} \end{equation}
\noindent Where $\bm{X}_{\hat{\mathcal{M}}_{\gamma_n}}$ denotes the $n \times d_n$ submatrix of $\bm{X}$ that is obtained by extracting its columns corresponding to the indices in $\hat{\mathcal{M}}_{\gamma_n}$. We additionally define $\bm{X}_{\mathcal{M}_{\gamma_n}}$ accordingly. Our initial estimator $\hat{\bm{\beta}_{I}}=(\hat{\beta}_{I,1},\ldots,\hat{\beta}_{I,d_n})$ is obtained using the Lasso. Let $\hat{\Sigma}_{\mathcal{M}_{\gamma_n}}=\bm{X}_{\mathcal{M}_{\gamma_n}}^{T}\bm{X}_{\mathcal{M}_{\gamma_n}}/n$, and let $\Sigma_{\mathcal{M}_{\gamma_n}}$ be its population counterpart. Our two stage estimator, $ \hat{\bm{\beta}}_{\hat{\mathcal{M}}_{\gamma_n}}$, is then formed by inserting zeroes corresponding to the covariates which were excluded in the screening step, and inserting the adaptive Lasso estimates, $\tilde{\bm{\beta}}_{\hat{\mathcal{M}}_{\gamma_n}}$, for covariates which were selected by the screening step. We need the following conditions for the combined two stage estimator to achieve sign consistency:
\newline\newline
\noindent \textbf{Condition J}: The matrix $\Sigma_{\mathcal{M}_{\frac{\gamma_n}{2}}}$ satisfies the restricted eigenvalue condition,\\
RE($s_n$,3)(see \cite{bickel09} for details):
\begin{equation} \phi_0=\min_{S \subseteq \{1,\ldots,d_n'\}, |S| \leq s_n} \min_{\bm{v}\neq 0, |\bm{v}_{S^c}| \leq 3|\bm{v}_{S}|} \frac{\bm{v}^{T}\Sigma_{\mathcal{M}_{\frac{\gamma_n}{2}}}\bm{v}}{\bm{v}^{T}\bm{v}} \geq c > 0,\end{equation}
where $\bm{v}=(v_1,\ldots,v_{d_n'})$ and $\bm{v}_{S}=(v_i,i \in S), \bm{v}_{S^c}=(v_i,i \in S^{c}).$
\newline\newline
\noindent \textbf{Condition K}: Let $\lambda_n$ and $\lambda_{I,n}$ be the regularization parameters of the adaptive lasso and the initial lasso estimator respectively. For some $\psi \in (0,1)$, we assume:
\begin{equation} cn^{1-\frac{\psi}{2}}(\frac{\phi_0}{s_n})^{3/2} \geq \lambda_{I,n} \geq \lambda_{n}n^{\psi/2} \end{equation}

\noindent \textbf{Condition L}: Let $\beta_{min}=\min_{i \leq s_n}|\beta_i|$, and $w_{\max}=\max_{i \leq s_n} w_i > 0$. Assume $\beta_{min}>\frac{2}{w_{\max}}$ and $\beta_{min} >2c\frac{\lambda_{I,n} s_n}{\phi_0 n}$.
\newline\newline
Condition J allows us to use the Lasso as our initial estimator. Notice that we placed the RE($s_n$,3) assumption on the matrix $\Sigma_{\mathcal{M}_{\frac{\gamma_n}{2}}}$, rather than the matrix $\hat{\Sigma}_{\hat{\mathcal{M}}_{\gamma_n}}$, given the indices in $\hat{M}_{\gamma_n}$ are random as a result of our screening procedure. Recall that for SIS, $\mathcal{M}_{\frac{\gamma_n}{2}}=\left\{1\leq i\leq p:|\rho_i| \geq \gamma_n/2 \right\}$, and $|\mathcal{M}_{\frac{\gamma_n}{2}}|=d_n'=O(d_n)$, and for GLSS we have a similar definition. Therefore, we are placing the RE($s_n$,3) assumption on the population covariance matrix of a fixed set of $d_n'$ predictors. Conditions K and L are standard assumptions, and are similar to the ones used in \cite{MM2016}. Condition K primarily places restrictions on the rate of increase of $\lambda_n$, and $\lambda_{I,n}$. Condition L places a lower bound on the magnitude of the non-zero parameters which decays with the sample size. 

The next theorem deals with the two stage SIS-Adaptive Lasso estimator. A very similar result applies to the two stage GLSS-Adaptive Lasso estimator, if we replace Conditions A,B,C (resp. D) with Conditions E,F,G,H (resp. I), to avoid repetition we omit the result. For the following theorem, the terms $\iota, \omega, K_{x,r}\textrm{, and } K_{\epsilon,q}$ have been defined in the paragraph preceding Theorem \ref{theorem2}, and $\tilde{\alpha}',\tilde{\alpha}$ have been defined in Theorem \ref{theorem4}.

\begin{thm}
\begin{enumerate}[(i)]
\item Assume Conditions A,B,C,J,K,L hold, then for $\gamma_n = c_3 n^{-\kappa}$ with $c_3 \leq c_1/2$ we have: 
\begin{align*} P(sgn( \hat{\bm{\beta}}_{\hat{\mathcal{M}}_{\gamma_n}})=sgn(\bm{\beta})) &\geq 1- O\left(s_np_n\left[\frac{n^{\omega} K_{x,r}^{r}}{(n/s_n)^{r/2-r\kappa/2}}- \exp(-\frac{n^{1-2\kappa}}{s_n^2K_{x,r}^4}) \right]\right)\nonumber \\
& - O\left(p_n\left[\frac{n^{\iota} K_{x,r}^{\tau}K_{\epsilon,q}^{\tau}}{n^{\tau-\tau\kappa}}- \exp(-\frac{n^{1-2\kappa}}{K_{x,r}^2K_{\epsilon,q}^2})\right]\right) \nonumber \\
& -O\left(d_n^{'2}\left[\frac{n^{\omega} K_{x,r}^{r}}{(n/s_n)^{r/2}}- \exp(-\frac{n}{s_n^2K_{x,r}^4})\right]\right) \nonumber \\
& -O\left(d_n'\left[\frac{n^{\iota} K_{x,r}^{\tau}K_{\epsilon,q}^{\tau}}{\lambda_n^{\tau}n^{\tau\psi/2}}+ \exp(-\frac{\lambda_n^2 n^{\psi-1}}{K_{x,r}^2K_{\epsilon,q}^2})\right]\right)\nonumber\end{align*}
\item Assume Conditions A,B,C,J,K,L hold, then for $\gamma_n = c_3 n^{-\kappa}$ with $c_3 \leq c_1/2$ we have: 
\begin{align*} P(sgn( \hat{\bm{\beta}}_{\hat{\mathcal{M}}_{\gamma_n}})=sgn(\bm{\beta})) &\geq 1-O(s_np_n\exp\left(-\frac{n^{1/2-\kappa}}{\upsilon_x^{2}s_n}\right)^{\tilde{\alpha}}) \nonumber \\
- &O(p_n\exp\left(-\frac{n^{1/2-\kappa}}{\upsilon_x\upsilon_{\epsilon}}\right)^{\tilde{\alpha}'}) - O(d_n^{'2}\exp\left(-\frac{n^{1/2}}{\upsilon_x^{2}s_n}\right)^{\tilde{\alpha}})\nonumber \\-& O(d_n'\exp\left(-\frac{\lambda_n n^{\psi/2-1/2} }{\upsilon_x\upsilon_{\epsilon}}\right)^{\tilde{\alpha}'}) \end{align*}
\end{enumerate}
\label{theorem7}
\end{thm}

To achieve sign consistency for the case of finite polynomial moments we require:

\noindent \textbf{Condition M}: Assume $\lambda_nn^{\psi/2-1/2} \rightarrow \infty$ and $p_n=o(\min(\frac{s_n(n/s_n)^{r/2-r\kappa/2}}{n^{\omega}},\frac{n^{\tau-\tau\kappa}}{n^{\iota}}))$, $d_n'=o(min((n/s_n)^{r/4-\omega/2},\lambda_n^{\tau}n^{\tau\psi/2-\iota}))$
\newline\newline
For the case of exponential moments, we require:
\newline\newline
\noindent \textbf{Condition N}: Assume $\lambda_nn^{\psi/2-1/2} \rightarrow \infty$, \\ $p_n=o(\min(\exp\left(\frac{Cn^{1/2-\kappa}}{s_n}\right)^{\tilde{\alpha}}/s_n, \exp(Cn^{1/2-\kappa})^{\tilde{\alpha}'}))$, \\ and $d_n'=o(\min(\exp\left(\frac{n^{1/2}}{s_n}\right)^{\tilde{\alpha}/2},\exp\left(\lambda_n n^{\psi/2-1/2} \right)^{\tilde{\alpha}'}))$
\newline\newline
\indent From Conditions M, N, and Theorem \ref{theorem6}, we see an additional benefit of using the two stage selection procedure as opposed to using the adaptive Lasso as a stand alone procedure. For example, if we assume $d_n \leq n^{2\kappa}\lambda_{\max}(\Sigma) =O(n)$, and that both the error and covariate processes are sub-Gaussian, we obtain $p_n=o(\exp(n^{\frac{1-2\kappa}{3}}))$ for the two stage estimator. By setting $d_n'=p_n$, we obtain the result when using the adaptive Lasso as a stand alone procedure, with the Lasso as its initial estimator. Under the scenario detailed above, the dimension of the feature space, which depends on $\lambda_n$ and $\psi$, for the stand alone adaptive Lasso can be at most $p_n=o(\exp(n^{\frac{1}{6}}))$. Therefore for $\kappa < 1/4$, we obtain a larger range for $p_n$ and a faster rate of decay using the two stage estimator. For $\kappa \geq 1/4$ it is not clear whether the two stage estimator has a larger range for $p_n$, compared to using the adaptive Lasso alone. 

The sign consistency of the stand alone adaptive Lasso estimator in the time series setting was established in \cite{MM2016}. Their result was obtained under strong mixing assumptions on the covariate and error processes, with the additional assumption that the error process is a martingale difference sequence. Additionally, in the ultrahigh dimensional setting they require a geometric decay rate on the strong mixing coefficients. In contrast, we obtain results for both the two stage and stand alone adaptive lasso estimator, and our results are obtained using the functional dependence measure framework. Besides assuming moment conditions, we are not placing any additional assumptions on the temporal decay of the covariate and error processes other than $\Delta_{0,q}(\bm{\epsilon})$,$\Phi_{0,q}(\bm{x}) < \infty$. Furthermore, we weaken the martingale difference assumption they place on the error process, thereby allowing for serial correlation in the error process. Finally, by using Nagaev type inequalities introduced in \cite{WuandWu2016}, our results are easier to interpret and also allow us obtain a higher range for $p_n$.

\section{Simulations}\label{sec:section 6}

In this section, we evaluate the performance of SIS, GLSS, and the two stage selection procedure using the adaptive Lasso. For GLSS instead of using the banded estimate for $\Sigma_k$ we use a tapered estimate: $\hat{\Sigma}_{k} * R_{l_n}$, where $\hat{\Sigma}_{k}= (\hat{\gamma}_{i-j,k})_{1\leq i,j \leq n}$ and
$R_{l_n}=(\max(1-\frac{|i-j|}{l_n},0))_{1\leq i,j \leq n}$ is the triangular kernel. We fix $l_n=15$, and we observed the results were fairly robust to the choice of $l_n$. In our simulated examples, we fix $n=200, s_n=6$ and $d_n=n-1$, while we vary $p_n$ from 1000 to 5000. We repeat each experiment 200 times. For screening procedures, we report the proportion of times the true model is contained in our selected model. For the two stage procedure using the adaptive Lasso, we report the proportion of times there was a $\lambda_n$ on the solution path which \emph{selected the true model}.
\newline\newline
\noindent\textbf{Case 1: Uncorrelated Features} 
\newline\newline
\indent Consider the model (\ref{eq1}), for the covariate process we have: 
\begin{equation} \bm{x}_t=A_1\bm{x}_{t-1} + \bm{\eta}_t \label{VAR} \end{equation} 
Where $A_1=diag(\gamma)$, and we vary $\gamma$ from .4 to .6. We set $\bm{\eta}_t \sim N(0,\Sigma_{\eta})$, or $\bm{\eta}_t \sim t_{5}(0,V)$ in which case the covariance matrix is $\Sigma_{\eta}= (5/3)*V$. For this scenario we will be dealing with uncorrelated predictors, we set $\Sigma_{\eta}=I_{p_n}$. For the error process, we have an AR(1) process: $\epsilon_i=\alpha\epsilon_{i-1}+e_i$. We let $\alpha$ vary from .6 to .9, and let $e_i \sim t_{5}$ or $e_i \sim N(0,1)$. We set $\bm{\beta}=(\bm{\beta}_1,\bm{\beta}_2)$, where $\bm{\beta}_1=(.5,.5,.5,.5,.5,.5)$ and $\bm{\beta}_2=\bm{0}$. Even though the features are uncorrelated, this is still a challenging setting, given the low signal to noise ratio along with heavy tails and serial dependence being present. 

The results are displayed in table \ref{table1}. The entries below ``Gaussian" correspond to the setting where both $e_i$ and $\bm{\eta}_i$ are drawn from a Gaussian distribution. Accordingly the entries under ``$t_5$" correspond to the case where $e_i$ and $\bm{\eta}_i$ are drawn from a $t_5$ distribution. We see from the results that the performance of SIS, and GLSS are comparable when $p_n=1000$, with moderate levels of temporal dependence, along with Gaussian covariates and errors. Interestingly, in this same setting, switching to heavy tails seems to have a much larger effect on the performance of SIS vs GLSS. In all cases, the performance of GLSS appears to be robust to the effects of serial correlation in the covariate and the error processes. Whereas, for SIS the performance severely deteriorates as we increase the level of serial correlation. For example, for our highest levels of serial correlation, SIS nearly always fails to contain the true model.

\begin{table}[]
\centering
\caption{Case 1}
\label{table1}
\begin{tabular}{|c|cccccc}
\hline
\multicolumn{1}{|l|}{} & \multicolumn{3}{c|}{SIS}                                                                   & \multicolumn{3}{c|}{GLSS}                                                                  \\ \hline
$(\gamma,\alpha)$      & \multicolumn{1}{c|}{(.4,.6)} & \multicolumn{1}{c|}{(.5,.8)} & \multicolumn{1}{c|}{(.6,.9)} & \multicolumn{1}{c|}{(.4,.6)} & \multicolumn{1}{c|}{(.5,.8)} & \multicolumn{1}{c|}{(.6,.9)} \\ \hline
Gaussian               & \multicolumn{1}{l}{}         & \multicolumn{1}{l}{}         & \multicolumn{1}{l}{}         & \multicolumn{1}{l}{}         & \multicolumn{1}{l}{}         & \multicolumn{1}{l|}{}         \\ \hline
$p_n=1000$                 & \multicolumn{1}{c|}{.95}     & \multicolumn{1}{c|}{.63}     & \multicolumn{1}{c|}{.15}     & \multicolumn{1}{c|}{.99}     & \multicolumn{1}{c|}{.99}     & \multicolumn{1}{c|}{.98}     \\ \hline
$p_n=5000$                 & \multicolumn{1}{c|}{.62}     & \multicolumn{1}{c|}{.11}     & \multicolumn{1}{c|}{.01}     & \multicolumn{1}{c|}{.95}     & \multicolumn{1}{c|}{.95}     & \multicolumn{1}{c|}{.97}     \\ \hline
$t_{5}$                & \multicolumn{6}{l|}{}                                                                                                                                                                   \\ \hline
$p_n=1000$                 & \multicolumn{1}{c|}{.58}     & \multicolumn{1}{c|}{.26}     & \multicolumn{1}{c|}{.06}     & \multicolumn{1}{c|}{.83}     & \multicolumn{1}{c|}{.84}     & \multicolumn{1}{c|}{.83}     \\ \hline
$p_n=5000$                 & \multicolumn{1}{c|}{.21}     & \multicolumn{1}{c|}{.01}     & \multicolumn{1}{c|}{0}       & \multicolumn{1}{c|}{.55}     & \multicolumn{1}{c|}{.49}     & \multicolumn{1}{c|}{.50}     \\ \hline
\end{tabular}
\end{table}

\noindent\textbf{Case 2: Correlated Features} 
\newline\newline
\indent We now compare the performance of SIS and GLSS for the case of correlated predictors. We have two scenarios:
\newline\newline
Scenario A: The covariate process is generated from (\ref{VAR}), with $A_1=diag(.4)$. $\bm{\eta}_t \sim N(0,\Sigma_{\eta})$, or $\bm{\eta}_t \sim t_{5}(0,V)$, with $\Sigma_{\eta}=\{.3^{|i-j|}\}_{i,j \leq p_n}$ for both cases. Therefore $\Sigma =\sum_{i=0}^{\infty} .4^{2i} \Sigma_{\eta}$. We set $\bm{\beta}_1=(1,-1,1,-1,1,-1)$ and $\bm{\beta}_2=\bm{0}$. We have an AR(1) process for the errors: $\epsilon_i=\alpha\epsilon_{i-1}+e_i$, we vary $\alpha$ from .4 to .8, and set $e_i \sim t_{5}$ or $e_i \sim N(0,1)$
\newline\newline
Scenario B: The covariate process is generated from (\ref{VAR}), with $A_1=\{.4^{|i-j|+1}\}_{i,j \leq p_n}$. And $\bm{\eta}_t \sim N(0,\Sigma_{\eta})$, or $\bm{\eta}_t \sim t_{5}(0,V)$, with $\Sigma_{\eta}=I_{p_n}$ for both cases. Therefore $\Sigma =\sum_{i=0}^{\infty} (A_{1}^{T})^{i}A_1^{i}$. We set $\bm{\beta}_1=(1,-1,1,-1,1,-1)$ and $\bm{\beta}_2=\bm{0}$. We have an AR(1) process for the errors: $\epsilon_i=\alpha\epsilon_{i-1}+e_i$, and we vary $\alpha$ from .4 to .8. The errors are generated in the same manner as in scenario A above.
\newline\newline
\indent The results are displayed in tables \ref{table2}, and \ref{table3} respectively. In scenario A, we have a Toeplitz covariance matrix for the predictors, and moderate levels of serial dependence in the predictors. The trends are similar to the ones we observed in case 1. The performance of SIS is sensitive to the effects of increasing the serial correlation in the errors, with the effect of serial dependence being more pronounced as we encounter heavy tail distributions. In contrast, increasing the level of serial dependence has a negligible impact on the performance of GLSS. For scenario B, we observe similar trends as in scenario A.

\begin{table}[]
\centering
\caption{Case 2: Scenario A}
\label{table2}
\begin{tabular}{|c|cccccc}
\hline
\multicolumn{1}{|l|}{} & \multicolumn{3}{c|}{SIS}                                                       & \multicolumn{3}{c|}{GLSS}                                                      \\ \hline
$\alpha$               & \multicolumn{1}{c|}{.4}  & \multicolumn{1}{c|}{.6}  & \multicolumn{1}{c|}{.8}  & \multicolumn{1}{c|}{.4}  & \multicolumn{1}{c|}{.6}  & \multicolumn{1}{c|}{.8}  \\ \hline
Gaussian               & \multicolumn{6}{l|}{}                                              \\ \hline
$p_n=1000$                 & \multicolumn{1}{c|}{.83} & \multicolumn{1}{c|}{.73} & \multicolumn{1}{c|}{.55} & \multicolumn{1}{c|}{.95} & \multicolumn{1}{c|}{.90} & \multicolumn{1}{c|}{.90} \\ \hline
$p_n=5000$                 & \multicolumn{1}{c|}{.38} & \multicolumn{1}{c|}{.30} & \multicolumn{1}{c|}{.07} & \multicolumn{1}{c|}{.63} & \multicolumn{1}{c|}{.63} & \multicolumn{1}{c|}{.57} \\ \hline
$t_{5}$                & \multicolumn{6}{l|}{}                                                                                                                                           \\ \hline
$p_n=1000$                 & \multicolumn{1}{c|}{.44} & \multicolumn{1}{c|}{.42} & \multicolumn{1}{c|}{.21} & \multicolumn{1}{c|}{.56} & \multicolumn{1}{c|}{.56} & \multicolumn{1}{c|}{.53} \\ \hline
$p_n=5000$                 & \multicolumn{1}{c|}{.01} & \multicolumn{1}{c|}{.04} & \multicolumn{1}{c|}{0}   & \multicolumn{1}{c|}{.16} & \multicolumn{1}{c|}{.14} & \multicolumn{1}{c|}{.16} \\ \hline
\end{tabular}
\end{table}

\begin{table}[]
\centering
\caption{Case 2: Scenario B}
\label{table3}
\begin{tabular}{|c|cccccc}
\hline
\multicolumn{1}{|l|}{} & \multicolumn{3}{c|}{SIS}                                                       & \multicolumn{3}{c|}{GLSS}                                                        \\ \hline
$\alpha$               & \multicolumn{1}{c|}{.4}  & \multicolumn{1}{c|}{.6}  & \multicolumn{1}{c|}{.8}  & \multicolumn{1}{c|}{.4}  & \multicolumn{1}{c|}{.6}   & \multicolumn{1}{c|}{.8}   \\ \hline
Gaussian               & \multicolumn{6}{l|}{}                                              \\ \hline
$p_n=1000$                 & \multicolumn{1}{c|}{.90} & \multicolumn{1}{c|}{.82} & \multicolumn{1}{c|}{.68} & \multicolumn{1}{c|}{.99} & \multicolumn{1}{c|}{1.00} & \multicolumn{1}{c|}{1.00} \\ \hline
$p_n=5000$                 & \multicolumn{1}{c|}{.71} & \multicolumn{1}{c|}{.64} & \multicolumn{1}{c|}{.26} & \multicolumn{1}{c|}{.95} & \multicolumn{1}{c|}{.97}  & \multicolumn{1}{c|}{.98}  \\ \hline
$t_{5}$                & \multicolumn{6}{l|}{}                                                                                                                                             \\ \hline
$p_n=1000$                 & \multicolumn{1}{c|}{.76} & \multicolumn{1}{c|}{.63} & \multicolumn{1}{c|}{.40} & \multicolumn{1}{c|}{.92} & \multicolumn{1}{c|}{.90}  & \multicolumn{1}{c|}{.92}  \\ \hline
$p_n=5000$                 & \multicolumn{1}{c|}{.37} & \multicolumn{1}{c|}{.26} & \multicolumn{1}{c|}{.06} & \multicolumn{1}{c|}{.76} & \multicolumn{1}{c|}{.74}  & \multicolumn{1}{c|}{.75}  \\ \hline
\end{tabular}
\end{table}

\noindent\textbf{Case 3: Two Stage Selection} 
\newline\newline
\indent We test the performance of the two stage GLSS-AdaLasso procedure. We also compare its performance with using the adaptive Lasso on its own. We use the Lasso as our initial estimator and select $\lambda_{I,n}$ using the modified BIC introduced in \cite{Wangetal2009}. \cite{Tang2013} extended the theory of the modified BIC to the case where $p>n$, $p=o(n^a), \textrm{ }a>1$, and independent observations. We conjecture that the same properties hold in a time series setting. We have two scenarios:
\newline\newline
Scenario A: The covariate process is generated from (\ref{VAR}), with $A_1=diag(.4)$. And $\bm{\eta}_t \sim N(0,\Sigma_{\eta})$, or $\bm{\eta}_t \sim t_{5}(0,V)$, with $(\Sigma_{\eta})_{i,j}=\{.8^{|i-j|}\}_{i,j \leq p_n}$. We set $\bm{\beta}_1=(.5,.5,.5,.5,.5,.5)$ and $\bm{\beta}_2=\bm{0}$. We have an AR(1) process for the errors: $\epsilon_i=\alpha\epsilon_{i-1}+e_i$, we vary $\alpha$ from .4 to .6, and set $e_i \sim t_{5}$ or $e_i \sim N(0,1)$
\newline\newline
Scenario B: The covariate process is generated from (\ref{VAR}), with $A_1=\{.4^{|i-j|+1}\}_{i,j \leq p_n}$. And $\bm{\eta}_t \sim N(0,\Sigma_{\eta})$, or $\bm{\eta}_t \sim t_{5}(0,V)$, with $(\Sigma_{\eta})_{i,j}=.8$ for $i \neq j$ and 1 otherwise. We set $\bm{\beta}_1=(.75,.75,.75,.75,.75,.75)$ and $\bm{\beta}_2=\bm{0}$. The errors are generated the same as in scenario A above.
\begin{table}[]
\centering
\caption{Case 3: Scenario A}
\label{table4}
\begin{tabular}{|c|c|c|c|c|c|c|}
\hline
\multicolumn{1}{|l|}{} & \multicolumn{3}{c|}{GLSS-AdaLasso} & \multicolumn{3}{c|}{AdaLasso} \\ \hline
$\alpha$               & .4         & .5        & .6        & .4       & .5       & .6      \\ \hline
Gaussian               & \multicolumn{6}{l|}{}                                              \\ \hline
$p_n=1000$                 & .79        & .65       & .49       & .60      & .49      & .35     \\ \hline
$p_n=5000$                 & .84        & .65       & .46       & .66      & .43      & .29     \\ \hline
$t_{5}$                & \multicolumn{6}{l|}{}                                              \\ \hline
$p_n=1000$                 & .45        & .37       & .23       & .32      & .22      & .14     \\ \hline
$p_n=5000$                 & .36        & .32       & .18       & .24        & .18      & .10     \\ \hline
\end{tabular}
\end{table}

\begin{table}[]
\centering
\caption{Case 3: Scenario B}
\label{table5}
\begin{tabular}{|c|c|c|c|c|c|c|}
\hline
\multicolumn{1}{|l|}{} & \multicolumn{3}{c|}{GLSS-AdaLasso} & \multicolumn{3}{c|}{AdaLasso} \\ \hline
$\alpha$               & .4         & .5        & .6        & .4       & .5       & .6      \\ \hline
Gaussian               & \multicolumn{6}{l|}{}                                              \\ \hline
$p_n=1000$                 & .86        & .72       & .59       & .57      & .49      & .34     \\ \hline
$p_n=5000$                 & .69        & .59       & .43       & .60      & .44      & .25     \\ \hline
$t_{5}$                & \multicolumn{6}{l|}{}                                              \\ \hline
$p_n=1000$                 & .48        & .41       & .22       & .30      & .19      & .10     \\ \hline
$p_n=5000$                 & .35        & .25       & .19       & .25      & .16      & .11     \\ \hline
\end{tabular}
\end{table}

In both scenarios we have a high degree of correlation between the predictors, low signal to noise ratio, along with mild to moderate levels of serial correlation in the covariate and error processes. The results are displayed in tables \ref{table4} and \ref{table5} for scenarios A and B respectively. We observe that the two stage estimator outperforms the standalone adaptive Lasso for both scenarios, with the difference being more pronounced in scenario B. For both scenarios, going from mild to moderate levels of serial correlation in the errors appears to significantly deteriorate the performance of the adaptive Lasso. This affects our results for the two stage estimator primarily at the second stage of selection. This sensitivity to serial correlation appears to increase as we encounter heavy tailed distributions.

\section{Real Data Example: Forecasting Inflation Rate}\label{sec:section 7}

In this section we focus on forecasting the 12 month ahead inflation rate. We use two major monthly price indexes as measures of inflation: the consumer price index (CPI), and the producer price index less finished goods (PPI). Specifically we are forecasting:
\begin{equation} y_{t+12}^{12}=100 \times \log\left(\frac{CPI_{t+12}}{CPI_{t}}\right), \textrm{ or } y_{t+12}^{12}=100 \times \log\left(\frac{PPI_{t+12}}{PPI_{t}}\right)  \end{equation}
Therefore the above quantities are approximately the percentage change in CPI or PPI over 12 months. Our data was obtained from the supplement to \cite{Jurado15}, and it consists of 132 monthly macroeconomic variables from January 1960 to December 2011, for a total of 624 observations. Apart from $\log(CPI)$ and $\log(PPI)$ which we are treating as $I(1)$, the remaining 130 macroeconomic time series have been transformed to achieve stationarity according to \cite{Jurado15}. Treating $\log(CPI)$, and $\log(PPI)$ as $I(1)$, has been found to provide an adequate description of the data according to \cite{Stock2002},\cite{Stock1999},\cite{MM2016}.

We consider forecasts from 8 different models. Similar to \cite{MM2016, Stock2002} our benchmark model is an AR(4) model: $\hat{y}_{t+12}^{12}=\hat{\alpha}_0+\sum_{i=0}^{3} \hat{\alpha}_i y_{t-i}$
%\begin{equation}  \hat{y}_{t+12}^{12}=\sum_{i=0}^{3} \hat{\alpha}_i y_{t-i} \end{equation}
, where $y_t=1200 \times log(CPI_t/CPI_{t-1})$ when forecasting CPI, and $y_t=1200 \times log(PPI_t/PPI_{t-1})$ when forecasting PPI. For comparison, we also consider an AR(4) model augmented with 4 factors. Specifically we have: 
\begin{equation} \hat{y}_{t+12}^{12}=\hat{\beta_0}+\sum_{i=0}^{3} \hat{\alpha}_i y_{t-i}+\hat{\bm{\gamma}} \bm{\hat{F}}_{t} \label{ar4} \end{equation}
Where $\bm{\hat{F}}_{t}$ are four factors which are estimated by taking the first four principal components of the 131 predictors along with three of their lags. We also consider forecasts estimated by the Lasso and the adaptive Lasso. And lastly we include forecasts estimated by the following two stage procedures: GLSS-Lasso, GLSS-adaptive Lasso, SIS-Lasso, and SIS-Adaptive Lasso. Our forecasting equation for the penalized regression and two stage forecasts is:
\begin{equation} y_{t+12}^{12}= \beta_0+\bm{x}_t\bm{\beta} + \epsilon_{t+12}^{12} \label{factors} \end{equation} 
Where $\bm{x}_t$ consists of $y_t$ and three of its lags along with the other 131 predictors and three of their lags, additionally we also include the first four estimated factors $\hat{F}_t$. Therefore $\bm{x}_t$ consists of 532 covariates in total. For each of the two stage methods, we set $d_n = \lceil n/\log(n) \rceil =73 $ for the first stage screening procedure. For the second stage selection, and the standalone lasso/adaptive lasso models, we select the tuning parameters and initial estimators using the approach described in section \ref{sec:section 6}.

We utilize a rolling window scheme, where the first simulated out of sample forecast was for January 2000 (2000:1). To construct this forecast, we use the observations between 1960:6 to 1999:1 (the first five observations are used in forming lagged covariates and differencing) to estimate the factors, and the coefficients. Therefore for the models described above, $t$=1960:6 to 1998:1. We then use the regressor values at $t$=1999:1 to form our forecast for 2000:1. Then the next window uses observations from 1960:7 to 1999:2 to forecast 2000:2. Using this scheme, in total we have 144 out of sample forecasts, and for each window we use $n=451$ observations for each regression model. The set-up described above allows us to simulate real-time forecasting.

Table \ref{table6} shows the mean squared error (MSE), and the mean absolute error (MAE) of the resulting forecasts relative to the MSE and MAE of the baseline AR(4) forecasts. We observe that the two stage GLSS methods clearly outperform the benchmark AR(4) model, and appear to have the best forecasting performance overall for both CPI and PPI, with the difference being more substantial when comparing by MSE. Furthermore GLSS-lasso and GLSS-adaptive Lasso do noticeably better than their SIS based counterparts with the differences being greater when forecasting PPI. We also note that the widely used factor augmented autoregressions do worse than the benchmark model AR(4) model.
\begin{table}[]
\centering
\caption{Inflation Forecasts: 12 month horizon}
\label{table6}
\begin{tabular}{|c|c|c|c|c|}
\hline
                    & CPI-MSE & CPI-MAE & PPI-MSE & PPI-MAE \\ \hline
AR(4)              & 1.00     & 1.00     & 1.00     & 1.00     \\ \hline
Lasso               & .94      & .99      & .69      & .89      \\ \hline
Adaptive Lasso      & 1.08     & 1.05     & .80      & .99      \\ \hline
SIS-Lasso           & .96      & .97      & .76      & .95      \\ \hline
SIS-Adaptive Lasso  & 1.03     & 1.00     & .82      & 1.00     \\ \hline
GLSS-Lasso          & .84      & .98      & .65      & .87      \\ \hline
GLSS-Adaptive Lasso & .94      & 1.00     & .70      & .92      \\ \hline
AR(4) + 4 Factors           & 1.18     & .99      & 1.08     & 1.09     \\ \hline
\end{tabular}
\end{table}

\section{Discussion}\label{sec:section 8}

In this paper we have analyzed the sure screening properties of SIS in the presence of dependence and heavy tails in the covariate and error processes. In addition, we have proposed a generalized least squares screening (GLSS) procedure, which utilizes the serial correlation present in the data when estimating our marginal effects. Lastly, we analyzed the theoretical properties of the two stage screening and adaptive Lasso estimator using the Lasso as our initial estimator. These results will allow practitioners to apply these techniques to many real world applications where the assumption of light tails and independent observations fails. 

There are plenty of avenues for further research, for example extending the theory of model-free screening methods such as distance correlation, or robust measures of dependence such as rank correlation to the setting where we have heavy tails and dependent observations. Other possibilities include extending the theory in this work, or to develop new methodology for long range dependent processes, or certain classes of non-stationary processes. Long range dependence, is a property which is prominent in a number of fields such as physics, telecommunications, econometrics, and finance (see \cite{Samorodnitsky2006} and references therein). If we assume the error process ($\epsilon_i$) is long range dependent, then by the proof of Theorem 1 in \cite{Wu2009} we have $\Delta_{0,q}(\bm{\epsilon})=\infty$. A similar result holds for the covariate process, therefore we may need to use a new dependence framework when dealing with long range dependent processes. Lastly, developing new methodology which aims to utilize the unique qualities of time series data such as serial dependence, and the presence of lagged covariates, would be a particularly fruitful area of future research.

\section{Appendix}
\subsection{Proofs of Results}

\begin{proof} [Proof of Theorem 1]
  $ $\newline

We first prove part (i), we start by obtaining a bound on:
\begin{equation}  P(|\hat{\rho}_j -\rho_j|>c_2 n^{-\kappa}) \end{equation}
Let $T_1=\sum_{t=1}^n X_{tj}^2/n$, $T_2=\sum_{t=1}^n X_{tj}Y_t/n$. Then $|\hat{\rho}_j-\rho_j| = |T_2/T_1-E(T_2)/E(T_1)|=\\
 |(T_1^{-1}-E(T_1)^{-1})(T_2-E(T_2))+(T_2-E(T_2))/E(T_1)+(T_1^{-1}-E(T_1)^{-1})E(T_2)|$
   $ $\newline

 \noindent Therefore:
 \begin{align} P(|\hat{\rho}_j -\rho_j|>c_2 n^{-\kappa})&\leq P(|(T_1^{-1}-E(T_1)^{-1})(T_2-E(T_2))|>c_2n^{-\kappa}/3)\label{triple1}\\
 &+P(|(T_2-E(T_2))/E(T_1)>c_2n^{-\kappa}/3|)\label{triple2}\\
 &+P(|(T_1^{-1}-E(T_1)^{-1})E(T_2)|>c_2n^{-\kappa}/3)\label{triple3}
 \end{align}
For the RHS of (\ref{triple1}), we obtain:
\begin{equation} (\ref{triple1}) \leq P(|(T_2-E(T_2))|>Cn^{-\kappa/2}) + P(|(T_1^{-1}-E(T_1)^{-1})|>Cn^{-\kappa/2}) \label{both}\end{equation}
Therefore it suffices to focus on terms (\ref{triple2}), (\ref{triple3}). For (\ref{triple2}), recall that

Recall that $T_2= \sum_{t=1}^n X_{tj}(\bm{x}_t\bm{\beta} + \epsilon_t)/n = \sum_{t=1}^n X_{tj}(\sum_{k=1}^{p_n}X_{tk}\beta_k + \epsilon_t)/n $. Now we let:
\begin{equation}
S_1= \sum_{t=1}^n X_{tj}(\sum_{k=1}^{p_n}X_{tk}\beta_k)/n \textrm{ and }S_2=\sum_{t=1}^n  X_{tj}\epsilon_t/n
\end{equation}
By Condition B, $E(X_{tj}\epsilon_t)=0$, therefore
\begin{equation} 
P(|T_2 -E(T_2)|>C n^{-\kappa}) \leq P(|S_1-E(S_1)| > C n^{-\kappa}/2) + P(|S_2| > Cn^{-\kappa})
\label{rhob}\end{equation}
Recall that $\sum_{k=1}^{p_n} \mathbbm{1}_{|\beta_k |>0}=s_n$, thus:
\begin{equation}  P\left(|S_1-E(S_1)| > \frac{c_2 n^{-\kappa}}{2}\right) \leq \sum_{k \in M_*}P\left(|\sum_{t=1}^n \frac{X_{tj}(X_{tk}\beta_k)}{n} - \beta_kE(X_{tj}X_{tk}) |> \frac{c_2 n^{-\kappa}}{2s_n} \right) \label{s1} \end{equation}
From section 2 in \cite{WuandWu2016}: $||X_{ij}||_{r} \leq  \Delta_{0,r}(\bm{X}_j) \leq \Phi_{0,r}(\bm{x})$. Using this we compute the cumulative functional dependence measure of $X_{tk}X_{tj}$ as: 
\begin{align}
\sum_{t=m}^{\infty} ||X_{tj}X_{tk}-X_{tj}^{*}X_{tk}^{*}||_{r/2} & \leq \sum_{t=m}^{\infty} (||X_{tj}||_{r}||X_{tk}-X_{tk}^{*}||_{r} + ||X_{tk}||_{r}||X_{tj}-X_{tj}^{*}||_{r})\nonumber \\
&\leq 2\Phi_{0,r}(\bm{x}) \Phi_{m,r}(\bm{x})= O(m^{-\alpha_x}) \label{cumx}
 \end{align}
Therefore we obtain: $\sup_m(m+1)^{\alpha_x}\sum_{t=m}^{\infty} ||X_{tj}X_{tk}-X_{tj}^{*}X_{tk}^{*}||_{r/2} \leq 2K_{x,r}^2$. Combining this with (\ref{s1}), and Theorem 2 in \cite{WuandWu2016}, yields:

\begin{align}
P\left(|S_1-E(S_1)| > \frac{c_2 n^{-\kappa}}{2}\right) \leq & Cs_n\left(\frac{n^{\omega} K_{x,r}^{r}}{(n/s_n)^{r/2-r\kappa/2}}+ \exp\left(-\frac{n^{1-2\kappa}}{s_n^2K_{x,r}^4}\right)\right)
\label{boundx}\end{align}
Similarly for $X_{tj}\epsilon_t$, by using Holder's inequality we obtain:
\begin{align}
\sum_{t=m}^{\infty} ||X_{tj}\epsilon_{t}-X_{tj}^{*}\epsilon_{t}^{*}||_{\tau} & \leq \sum_{t=m}^{\infty} (||X_{tj}||_{r}||\epsilon_{t}-\epsilon_{t}^{*}||_{q} + ||\epsilon_{t}||_{q}||X_{tj}-X_{tj}^{*}||_{r})\nonumber \\
&\leq \Delta_{0,q}(\bm{\epsilon})\Phi_{m,r}(\bm{x})+\Delta_{m,q}(\bm{\epsilon})\Phi_{0,r}(\bm{x})= O(m^{-\alpha}) \label{cume}
 \end{align}
Therefore $\sup_m (m+1)^{\alpha}\sum_{t=m}^{\infty} ||X_{tj}\epsilon_{t}-X_{tj}^{*}\epsilon_{t}^{*}||_{\tau} \leq 2K_{x,r}K_{\epsilon,q}$. Using Theorem 2 in \cite{WuandWu2016}, we obtain:
\begin{equation}
P\left(|S_2| > \frac{c_2 n^{-\kappa}}{2}\right) \leq O\left(\frac{n^{\iota} K_{x,r}^{\tau}K_{\epsilon,q}^{\tau}}{n^{\tau-\tau\kappa}}+ \exp\left(-\frac{n^{1-2\kappa}}{K_{x,r}^2K_{\epsilon,q}^2}\right)\right)
\label{bounde}\end{equation}

For (\ref{triple3}), assuming $E(X_{ij}^2)=O(1) \textrm{ } \forall j \leq p_n$, and $ \max_{j \leq p_n}E(X_{tj}Y_t) < L<\infty$ we obtain: 
\begin{equation}(\ref{triple3}) \leq P(|T_1-E(T_1)|>T_1Cn^{-\kappa}) \leq P(|T_1-E(T_1)|>MCn^{-\kappa}) + P(T_1<M)\label{inverse} \end{equation}
We set $M < \min_{j\leq p_n} E(X_{ij}^2)-\epsilon$, for $\epsilon>0$. We then have:
\begin{equation} P(T_1<M) \leq P(|T_1-E(T_1)| > E(T_1)-M) \end{equation}
We can then bound the above two equations similar to (\ref{boundx}).
By combining (\ref{both})(\ref{rhob}),(\ref{boundx}),(\ref{bounde}),(\ref{inverse}), along with union bound we obtain:
\begin{align*} P\left(\max_{j \leq p_n}|\hat{\rho}_j -\rho_j|>c_2 n^{-\kappa}\right)
&\leq O\left(s_np_n\left[\frac{n^{\omega} K_{x,r}^{r}}{(n/s_n)^{r/2-r\kappa/2}}+ \exp\left(-\frac{n^{1-2\kappa}}{s_n^2K_{x,r}^4}\right) \right]\right)\nonumber \\
& + O\left(p_n\left[\frac{n^{\iota} K_{x,r}^{\tau}K_{\epsilon,q}^{\tau}}{n^{\tau-\tau\kappa}}+ \exp(-n^{1-2\kappa}/K_{x,r}^2K_{\epsilon,q}^2)\right]\right)
\end{align*}

To prove part (ii), we follow the steps in the proof of Theorem 2 in \cite{LiGetal2012}. Let $\mathcal{A}_n=\{\max_{k \in M_*}|\hat{\rho}_k-\rho_k|\leq \frac{c_1 n^{-\kappa}}{2}\}$. On the set $\mathcal{A}_n$, by Condition \textbf{A}, we have: 
\begin{equation} |\hat{\rho}_k| \geq |\rho_k|-|\hat{\rho}_k - \rho_k| \geq c_1 n^{-\kappa}/2, \textrm{  }\forall k \in M_*\end{equation}
Hence by our choice of $\gamma_n$, we obtain $P\left(\mathcal{M}_*\subset \hat{\mathcal{M}}_{\gamma_n} \right) > P(\mathcal{A}_n)$. By applying part (i), the result follows.
\newline\newline
 For part (iii) we follow the steps in the proof of Theorem 3 in \cite{LiGetal2012}. Using $Var(Y_t), Var(X_{tj})=O(1) \textrm{ for } j \leq p_n$, along with Condition B, we obtain $\sum_{k=1}^{p_n} \rho_k^2=O(\lambda_{\max}(\Sigma))$. Then on the set $\mathcal{B}_n=\{\max_{k \leq p_n}|\hat{\rho}_k-\rho_k|\leq c_4 n^{-\kappa}\}$, the number of $\{k: |\hat{\rho}_k| > 2c_4n^{-\kappa}\}$ cannot exceed the number of $\{k: |\rho_k| > c_4n^{-\kappa}\}$ which is bounded by $O(n^{2\kappa}\lambda_{\max}(\Sigma))$. Therefore, by setting $c_4=c_{3}/2$ we obtain:
\begin{equation} P\left(|\hat{\mathcal{M}}_{\gamma_n}| <  O(n^{2\kappa}\lambda_{\max}(\Sigma)) \right) > P(\mathcal{B}_n) \end{equation}
The result then follows from part (i).

\end{proof}

\begin{proof} [Proof of Theorem \ref{theorem4}]
  $ $\newline

We follow the steps from the proof of Theorem \ref{theorem2}. Let $\bm{T}=(T_{1},\ldots,T_n)$ where $T_i=X_{ij}X_{ik}$, and let $\bm{R}=(R_{1},\ldots,R_n)$ where $R_i= X_{ij}\epsilon_i$. We need to bound the sums: $\sum_{i=1}^n (T_i-E(T_i))/n$ and $\sum_{i=1}^n R_i/n$.

By Theorem 1 in \cite{Wu2005}, $\Theta_{q}(\bm{T}) \leq \Delta_{0,q}(\bm{T})$, and from Section 2 in \cite{WuandWu2016}: $||X_{ij}||_{q} \leq  \Delta_{0,q}(\bm{X}_j) \leq \Phi_{0,q}(\bm{x})$. Additionally, by Holders inequality we have
\begin{equation}\Delta_{0,q}(\bm{T}) \leq \sum_{t=0}^{\infty} (||X_{tj}||_{2q}||X_{tk}-X_{tk}^{*}||_{2q} + ||X_{tk}||_{2q}||X_{tj}-X_{tj}^{*}||_{2q}) \leq 2\Phi_{0,2q}^2(\bm{x})  \end{equation}
Using these, along with Condition D we obtain:
\begin{equation} \sup_{q \geq 4} q^{-2\tilde{\alpha}_x}\Theta_{q}(\bm{T}) \leq \sup_{q \geq 4} q^{-2\tilde{\alpha}_x}\Delta_{0,q}(\bm{T}) \leq \sup_{q \geq 4} 2q^{-2\tilde{\alpha}_x}\Phi_{0,2q}^{2}(\bm{x}) < \infty \end{equation}
Combining the above and using Theorem 3 in \cite{WuandWu2016}, we obtain:
\begin{equation} P\left(\bigg|\sum_{i=1}^nT_i-E(T_i)\bigg| > \frac{c_2 n^{1-\kappa}}{2}\right) \leq  C\exp\left(-\frac{n^{1/2-\kappa}}{\upsilon_x^2}\right)^{\tilde{\alpha}} \end{equation}
Similarly, using the same procedure we obtain:
\begin{equation}P\left(\bigg|\sum_{i=1}^nR_i\bigg| > \frac{c_2 n^{1-\kappa}}{2}\right) \leq C\exp\left(-\frac{n^{1/2-\kappa}}{\upsilon_x\upsilon_{\epsilon}}\right)^{\tilde{\alpha}'} \end{equation}
Now using the above bounds and following the steps in the proof of Theorem \ref{theorem2} we obtain the results. 

\end{proof}

\begin{proof} [Proof of Lemma 1]
  $ $\newline

By the proof of Theorem 2 in \cite{Wu2009}, we have:
\begin{equation}||\hat{\Sigma}_{k,l_n}-\Sigma_{k}||_2 \leq 2\sum_{i=1}^{l_n}|\hat{\gamma}_{i,k}-\gamma_{i,k}|+2\sum_{i=l_n+1}^{\infty}|\gamma_{i,k}| \label{split1}\end{equation}
Recall that $\hat{\rho}_k$ is the OLS estimate of the marginal projection, by (\ref{marg}) we have $\hat{\epsilon}_{t,k}=\epsilon_{t,k}-X_{tk}(\hat{\rho}_k-\rho_k)=\epsilon_{t,k}-X_{tk}(\frac{\sum_{j=1}^{n} X_{jk}\epsilon_{j,k}/n}{\sum_{j=1}^{n} X_{jk}^2/n})$. Which gives us:
\begin{align}\hat{\gamma}_{i,k}= \frac{1}{n}&\sum_{t=1}^{n-|i|} \Biggl[\epsilon_{t,k}\epsilon_{t+|i|,k} - \epsilon_{t,k} X_{t+|i|,k}\bigg(\sum_{j=1}^{n} X_{jk}\epsilon_{j,k}/n\bigg) \\
&-\epsilon_{t+|i|,k} X_{tk}\bigg(\sum_{j=1}^{n} X_{jk}\epsilon_{j,k}/n\bigg)+X_{tk}X_{t+|i|,k}\bigg(\sum_{j=1}^{n} X_{jk}\epsilon_{j,k}/n\bigg)^2\Biggr] \end{align}
By Condition E and $l_n=c\log(n)$, for sufficiently large $c$, we have: $\sum_{i=l_n+1}^{\infty}|\gamma_{i,k}|=o(n^{-\kappa})$, so we focus on the term $\sum_{i=1}^{l_n}|\hat{\gamma}_{i,k}-\gamma_{i,k}|$ in (\ref{split1}). We then have:
\begin{equation} P(||\hat{\Sigma}_{k,l_n}-\Sigma_{k}||_2 > cn^{-\kappa}) \leq \sum_{i=1}^{l_n} P\left(|\hat{\gamma}_{i,k}-\gamma_{i,k}|>cn^{-\kappa}/l_n\right) \label{main1}\end{equation}
And
\begin{align} P(|\hat{\gamma}_{i,k}-\gamma_{i,k}|>&\frac{cn^{-\kappa}}{l_n}) \leq P\Biggl(\bigg|\frac{1}{n}\sum_{t=1}^{n-|i|} \epsilon_{t,k}\epsilon_{t+|i|,k}-E(\frac{1}{n}\sum_{t=1}^{n-|i|} \epsilon_{t,k}\epsilon_{t+|i|,k})\bigg| \label{pre}\\
&+\bigg|E(\frac{1}{n}\sum_{t=1}^{n-|i|} \epsilon_{t,k}\epsilon_{t+|i|,k})-\gamma_{i,k}\bigg|>cn^{-\kappa}/4l_n\Biggr)\label{A}\\
&+ P\left(\bigg|\frac{1}{n}\sum_{t=1}^{n-|i|}\epsilon_{t,k} X_{t+|i|,k}\left(\frac{\sum_{j=1}^{n} X_{jk}\epsilon_{j,k}/n}{\sum_{j=1}^{n} X_{jk}^2/n}\right)\bigg|> cn^{-\kappa}/4l_n\right)\label{B}\\
&+P\left(\bigg|\frac{1}{n}\sum_{t=1}^{n-|i|}\epsilon_{t+|i|,k} X_{tk}\left(\frac{\sum_{j=1}^{n} X_{jk}\epsilon_{j,k}/n}{\sum_{j=1}^{n} X_{jk}^2/n}\right)\bigg|> cn^{-\kappa}/4l_n\right)\label{C}\\
&+P\left(\bigg|\frac{1}{n}\sum_{t=1}^{n-|i|}X_{tk}X_{t+|i|,k}\left(\frac{\sum_{j=1}^{n} X_{jk}\epsilon_{j,k}/n}{\sum_{j=1}^{n} X_{jk}^2/n}\right)^2\bigg|>cn^{-\kappa}/4l_n\right)\label{D}
\end{align}
For (\ref{A}), the bias $ |E(\sum_{t=1}^{n-|i|} \frac{\epsilon_{t,k}\epsilon_{t+|i|,k}}{n}-\gamma_{i,k}| \leq \frac{i\gamma_{i,k}}{n} $. Using the techniques in the proof of Theorem \ref{theorem2} we can then bound (\ref{pre}). For (\ref{B}) we have:
\begin{equation} (\ref{B}) \leq P\left(\frac{|\sum_{j=1}^{n} X_{jk}\epsilon_{j,k}/n|}{\sum_{j=1}^{n} X_{jk}^2/n}> cn^{-\kappa}/Ml_n\right) + P\left(\bigg|\frac{1}{n}\sum_{t=1}^{n-|i|}\epsilon_{t,k} X_{t+|i|,k}\bigg|>M\right) \end{equation}
And $P\left(\left|\frac{1}{n}\sum_{t=1}^{n-|i|}\epsilon_{t,k} X_{t+|i|,k}\right|>M\right)$
\begin{equation} \leq P\left(\bigg|\frac{1}{n}\sum_{t=1}^{n-|i|}\epsilon_{t,k} X_{t+|i|,k}-E(\frac{1}{n}\sum_{t=1}^{n-|i|}\epsilon_{t,k} X_{t+|i|,k})\bigg|> M-\bigg|E(\frac{1}{n}\sum_{t=1}^{n-|i|}\epsilon_{t,k} X_{t+|i|,k})\bigg|\right) \end{equation}
And we set $M> \max_{k \leq p_n}\max_{i \leq l_n}2|E(\epsilon_{t,k}X_{t+|i|,k})|+\epsilon$, for some $\epsilon > 0$. Similarly we have $P(\frac{|\sum_{j=1}^{n} X_{jk}\epsilon_{j,k}/n|}{\sum_{j=1}^{n} X_{jk}^2/n}> cn^{-\kappa}/Ml_n)$
\begin{equation} \leq P\left(\bigg|\sum_{j=1}^{n} X_{jk}\epsilon_{j,k}/n\bigg|> M_1Cn^{-\kappa}/l_n\right) +  P\left(\sum_{j=1}^{n} X_{jk}^2/n < M_1\right)\end{equation}
Where we set $M_1< \min_{j\leq p_n} E(X_{ij}^2)-\epsilon$, for $\epsilon>0$. The same method we used for (\ref{pre}) can be applied to (\ref{C}), (\ref{D}). Using the techniques in the proof of Theorem \ref{theorem2}, and (\ref{main1}), we obtain the result. For (ii), we follow the same procedure as in (i), and apply the methods seen in the proof of Theorem \ref{theorem4}.

\end{proof}

\begin{proof} [Proof of Theorem \ref{theorem6}]
  $ $\newline

For (i), as before we start with a bound on: $P(|\hat{\beta}_k^{M} -\beta_k^{M}|>c_2 n^{-\kappa})$. Using Condition E, we can write:
 \begin{equation} \beta_k^{M}=(E(\bm{X}_k^{T}\Sigma_k^{-1}\bm{X}_k)/n)^{-1} E(\bm{X}_k^T\Sigma_k^{-1}\bm{y}^k/n)+O(1/n) \nonumber \end{equation}
 After combining this with (\ref{marg}), it suffices to obtain a bound for:
\begin{equation} P(|(\bm{X}_k^{T}\hat{\Sigma}_{k,l_n}^{-1}\bm{X}_k/n)^{-1} \bm{X}_k^T\hat{\Sigma}_{k,l_n}^{-1}\bm{\epsilon}^k/n-(E(\bm{X}_k^{T}\Sigma_k^{-1}\bm{X}_k))^{-1} E(\bm{X}_k^T\Sigma_k^{-1}\bm{\epsilon}^k)|>c n^{-\kappa}) \end{equation}
Similar to the proof of Theorem \ref{theorem2} we let $T_1=\bm{X}_k^{T}\hat{\Sigma}_{k,l_n}^{-1}\bm{X}_k/n$, \\$T_2= \bm{X}_k^T\hat{\Sigma}_{k,l_n}^{-1}\bm{\epsilon}^k/n$, $T_3=E(\bm{X}_k^{T}\Sigma_k^{-1}\bm{X}_k)$, and $T_4=E(\bm{X}_k^T\Sigma_k^{-1}\bm{\epsilon}^k)$. Then:
\begin{align} 
|\hat{\beta}_k^{M} -\beta_k^{M}| = |T_2/T_1-T_4/T_3|&= |(T_1^{-1}-T_3^{-1})(T_2-T_4) \nonumber \\ 
&+(T_2-T_4)/T_3+(T_1^{-1}-T_3^{-1})T_4| \label{splitgls} \end{align}
Following the steps in the proof of Theorem \ref{theorem2}, it suffices to focus on the terms: 
\begin{equation}
P(|T_1-T_3|>cn^{-\kappa}) \textrm{ and } P(|T_2-T_4|>cn^{-\kappa}) 
\end{equation}
We then have:
\begin{align} P(|T_2-T_4|>C n^{-\kappa}) &\leq P(|\bm{X}_k^T(\hat{\Sigma}_{k,l_n}^{-1}-\Sigma_k^{-1})\bm{\epsilon}^k/n|>C n^{-\kappa}/2)\label{combine}\\
&+ P(|\bm{X}_k^T\Sigma_k^{-1}\bm{\epsilon}^k/n-E(\bm{X}_k^T\Sigma_k^{-1}\bm{\epsilon}^k)|>C n^{-\kappa}/2)  \nonumber \end{align}
We first deal with the term $\bm{X}_k^T\Sigma_k^{-1}\bm{\epsilon}^k/n$. We can rewrite this term as $\tilde{\bm{X}}_k^T\tilde{\bm{\epsilon}}^k/n$, where $\tilde{\bm{X}}_k=V_k\bm{X}_k, \tilde{\bm{\epsilon}}^k=V_k\bm{\epsilon}^k/n$, $V_k$ is a lower triangle matrix and the square root of $\Sigma_k^{-1}$. Ignoring the first $L_k$ observations, we can express: 
\begin{equation}\tilde{\bm{X}}_k^T\tilde{\bm{\epsilon}}^k/n=\sum_{t=L_k+1}^{n} \left(\epsilon_{t,k}-\sum_{i=1}^{L_k} \alpha_{i,k}\epsilon_{t-i,k}\right)\left(X_{t,k}-\sum_{i=1}^{L_k} \alpha_{i,k}X_{t-i,k}\right)\end{equation}
, where $(\alpha_{1,k},\ldots,\alpha_{L_k,k})$ are the autoregressive coefficients of the process $\epsilon_{t,k}$.

We compute the cumulative functional dependence measure of $\tilde{X}_{t,k}\tilde{\epsilon}_{t,k}$ as: 
\begin{align}
\sum_{l=m}^{\infty} ||\tilde{X}_{l,k}\tilde{\epsilon}_{l,k}-\tilde{X}_{l,k}^{*}\tilde{\epsilon}_{l,k}^{*}||_{\tau'} & \leq \sum_{l=m}^{\infty} (||\tilde{X}_{l,k}||_{r}||\tilde{\epsilon}_{l,k}-\tilde{\epsilon}_{l,k}^{*}||_{q'} + ||\tilde{\epsilon}_{l,k}||_{q'}||\tilde{X}_{l,k}-\tilde{X}_{l,k}^{*}||_{r}) 
\end{align}
We have: $||\tilde{X}_{l,k}-\tilde{X}_{l,k}^{*}||_{r} \leq ||X_{l,k}-X_{l,k}^{*}||_{r}+\sum_{i=1}^{L_k}|\alpha_i| ||X_{k,l-i}-X_{k,l-i}^{*}||_r$. And by our assumptions $||\tilde{\epsilon}_{l,k}-\tilde{\epsilon}_{l,k}^{*}||_{q'}=0$, for $l>0$. From which we obtain:
\begin{align}
\sum_{l=m}^{\infty} ||\tilde{X}_{l,k}\tilde{\epsilon}_{l,k}-\tilde{X}_{l,k}^{*}\tilde{\epsilon}_{l,k}^{*}||_{\tau'}  & \leq C\Phi_{m,r}=O(m^{-\alpha_x})
\end{align}
Using Theorem 2 in \cite{WuandWu2016}:
\begin{equation}  P(|\bm{X}_k^T\Sigma_k^{-1}\bm{\epsilon}^k/n-E(\bm{X}_k^T\Sigma_k^{-1}\bm{\epsilon}^k)|>Cn^{-\kappa}) \leq O\left(\frac{n^{\iota}K_{x,r}^{\tau'}}{n^{\tau'-\tau'\kappa}}\right) \label{first}\end{equation}
For the term $|\bm{X}_k^T(\hat{\Sigma}_{k,l_n}^{-1}-\Sigma_k^{-1})\bm{\epsilon}^k/n|$, using Cauchy-Schwarz inequality:
\begin{equation} \frac{|\bm{X}_k^T(\hat{\Sigma}_{k,l_n}^{-1}-\Sigma_k^{-1})\bm{\epsilon}^k/n|}{||\bm{X}_k||_2||\bm{\epsilon}^k||_2} \leq \frac{||(\hat{\Sigma}_{k,l_n}^{-1}-\Sigma_k^{-1})\bm{\epsilon}^k||_2}{n||\bm{\epsilon}^k||_2} \leq \frac{||\hat{\Sigma}_{k,l_n}^{-1}-\Sigma_k^{-1}||_2}{n}\label{cauchy} \end{equation}
Using (\ref{cauchy}) we obtain:
\begin{align}  P(|\bm{X}_k^T(\hat{\Sigma}_{k,l_n}^{-1}-\Sigma_k^{-1})\bm{\epsilon}^k/n|>C n^{-\kappa}) &\leq P(||\bm{X}_k||_2||\bm{\epsilon}^k||_2||\hat{\Sigma}_{k,l_n}^{-1}-\Sigma_k^{-1}||_2/n> Cn^{-\kappa}) \label{second} \end{align}
Where the right hand side of (\ref{second}) is:
\begin{align} \leq P(||\hat{\Sigma}_{k,l_n}^{-1}-\Sigma_{k}^{-1}||_2 > Cn^{-\kappa}/\sqrt{M}) + P\left(\bigg(\sum_{i=1}^n X_{ik}^2/n\bigg) \bigg(\sum_{i=1}^n\epsilon_{i,k}^2/n\bigg) > M\right) \label{doublesquare}\end{align}
Let $M=M_1M_2$, where $M_1 \geq \max_{k \leq p_n} E(X_{i,k}^2)+\epsilon$, and $M_2= \max_{k \leq p_n} E(\epsilon_{i,k}^2)+\epsilon$, for some $\epsilon>0$. The second term of (\ref{doublesquare}) is:
\begin{align} &\leq P\left(\sum_{i=1}^n X_{ik}^2/n > M_1\right) + P\left(\sum_{i=1}^n\epsilon_{i,k}^2/n > M_2\right) \label{split3} \end{align}  
We can bound the above using the same techniques as in the previous proofs. 

By Condition E, the spectral density of the process $\epsilon_{t,k}, \forall k\leq p_n$ is bounded away from zero and infinity. Therefore, $0 < C_1\leq \lambda_{min}(\Sigma_k) \leq \lambda_{max}(\Sigma_k)\leq C_2 < \infty, \forall k \leq p_n$ \cite{Wu2009}. We then use:
\begin{align} \lambda_{\min}(\Sigma_{k})||\hat{\Sigma}_{k,l_n}^{-1}-\Sigma_{k}^{-1}||_2  &\leq ||\Sigma_k^{\frac{1}{2}}(\hat{\Sigma}_{k,l_n}^{-1}-\Sigma_{k}^{-1})\Sigma_k^{\frac{1}{2}}||_2\nonumber\\
& = ||\Sigma_k^{\frac{1}{2}}\hat{\Sigma}_{k,l_n}^{-1}\Sigma_k^{\frac{1}{2}}-I_n||_2  \label{decomp4}\end{align} 

Let $a_1 \geq a_2 \geq \ldots \geq a_n$ be the ordered eigenvalues of $\Sigma_k^{-\frac{1}{2}}\hat{\Sigma}_{k,l_n}\Sigma_k^{-\frac{1}{2}}$, therefore $||\Sigma_k^{\frac{1}{2}}\hat{\Sigma}_{k,l_n}^{-1}\Sigma_k^{\frac{1}{2}}-I_n||_2=\max_i |\frac{1}{a_i}-1|=\max_i |\frac{a_i-1}{a_i}|$. We then have
 \begin{align}\max_i |a_i-1|=||\Sigma_k^{-\frac{1}{2}}\hat{\Sigma}_{k,l_n}\Sigma_k^{-\frac{1}{2}}-I_n||_2 \leq \lambda_{\max}(\Sigma_k^{-1})||\hat{\Sigma}_{k,l_n}-\Sigma_{k}||_2 \label{decomp5}\end{align}
Let $a_j=argmin_{a_i}|a_i^{-1}|$, using this and (\ref{decomp4}),(\ref{decomp5}) we obtain:
\begin{align} P(||\hat{\Sigma}_{k,l_n}^{-1}-\Sigma_{k}^{-1}||_2 > Cn^{-\kappa}) &\leq P(||\hat{\Sigma}_{k,l_n}-\Sigma_{k}||_2 > Ca_jn^{-\kappa})\nonumber\\
& \leq P(||\hat{\Sigma}_{k,l_n}-\Sigma_{k}||_2 > CM_3n^{-\kappa}) +P(|a_j| < M_3) \label{decomp6}\end{align}
Where $M_3 \in (0,1-\epsilon)$ for $\epsilon> 0$. We then have
\begin{equation} P(|a_j| < M_3) \leq P(|a_j-1| > 1-M_3) \leq P(||\hat{\Sigma}_{k,l_n}-\Sigma_{k}||_2 >1-M_3)\nonumber\end{equation}
Combining the above with (\ref{decomp6}) and Lemma \ref{lem2}, we obtain:
\begin{equation}P(||\hat{\Sigma}_{k,l_n}^{-1}-\Sigma_{k}^{-1}||_2 > Cn^{-\kappa}) \leq l_nO\left(\frac{n^{\iota}l_n^{\tau'}K_{x,r}^{\tau'}\tilde{K}_{\epsilon,q'}^{\tau'}}{n^{\tau'-\tau'\kappa}}+\frac{n^{\zeta}l_n^{q'/2}\tilde{K}_{\epsilon,q'}^{q'}}{n^{q'/2-q'\kappa/2}} +\frac{n^{\omega}l_n^{r/2}K_{x,r}^{r}}{n^{r/2}}\right)  \label{corr}\end{equation}

By (\ref{combine}),(\ref{first}),(\ref{second}),(\ref{split3}),(\ref{corr}) we obtain a bound for $P(|T_2-E(T_2)|>C n^{-\kappa})$. For the term $P(|T_1-E(T_1)|>C n^{-\kappa})$, we proceed in a similar fashion:
\begin{align} P(|T_1-E(T_1)|>C n^{-\kappa}) &\leq P(|\bm{X}_k^T(\hat{\Sigma}_{k,l_n}^{-1}-\Sigma_k^{-1})\bm{X}^k/n|>C n^{-\kappa}/2)\label{combine2}\nonumber\\
&+ P(|\bm{X}_k^T\Sigma_k^{-1}\bm{X}^k/n-E(\bm{X}_k^T\Sigma_k^{-1}\bm{X}^k)|>C n^{-\kappa}/2)  \nonumber \end{align}
We can then obtain a bound on the above terms by following a similar procedure as before. Combining these gives us the result for (i). For (ii), using the result from (i) we follow a similar procedure to the proof of Theorem \ref{theorem2}. For (iii) and (iv) we follow the same procedure as (i) and (ii), and apply the methods seen in the proof of Theorem \ref{theorem4}; we omit the details.

\end{proof}

\begin{proof} [Proof of Corollary \ref{cor1}]
  $ $\newline

Recall that:
\begin{equation} \beta_{k}^M=E(y_t-\sum_{i=1}^{L_k} \alpha_iy_{t-i})(X_{t,k}-\sum_{i=1}^{L_k} \alpha_iX_{t-i,k})/(E(X_{t,k}-\sum_{i=1}^{L_k} \alpha_iX_{t-i,k})^2)\nonumber\end{equation}
Therefore by our assumption, we have that $\beta_k^M \propto \rho_k$ whenever $\beta_k^M>0$. Using this we obtain $\sum_{k=1}^{p_n} (\beta_k^{M})^2 = O(\sum_{k=1}^{p_n} \rho_k^2)= O(\lambda_{\max}(\Sigma))$. We obtain the result, by following the procedure in the proof of Theorem \ref{theorem2} and using the results from Theorem \ref{theorem6}.

\end{proof}

\begin{proof} [Proof of Theorem \ref{theorem7}]
  $ $\newline

For simplicity we only prove part (i), the proof for part (ii) follows similarly. We will work on the following set $\mathcal{D}_n=\mathcal{A}_n \cap \mathcal{B}_n \cap \mathcal{C}_n$, where
\begin{align*} 
&\mathcal{A}_n=\{max_{k \leq p_n}|\hat{\rho}_k-\rho_k|\leq c_3n^{-\kappa}/2\}  \nonumber\\
& \mathcal{B}_n=  \{max_{i,j \leq d_n'}|[\Sigma_{\mathcal{M}_{\frac{\gamma_n}{2}}}-\hat{\Sigma}_{\mathcal{M}_{\frac{\gamma_n}{2}}}]_{i,j}|\leq \frac{\phi_0}{16s_n} \} \nonumber\\
&\mathcal{C}_n=\{max_{k \leq d_n'}|\sum_{i=1}^{n} X_{ik}\epsilon_i| \leq \lambda_nn^{\psi/2} \}
\end{align*}

On the set $\mathcal{A}_n$, if we apply screening as a first stage procedure, by our choice of $\gamma_n$, we obtain:
\begin{equation}
\mathcal{M}_{*} \subset \hat{\mathcal{M}}_{\gamma_n} \subset \mathcal{M}_{\gamma_n/2}
\label{inclusion} \end{equation}
Next we need to use Lemma 7 and 8 in \cite{MM2016}, specifically we need to show our reduced model satisfies conditions DGP 3,DESIGN, and WEIGHTS in \cite{MM2016}. On the set $ \mathcal{B}_n$, by Lemma 1 in \cite{MM2016}, we have $\phi_{\Sigma_{\hat{\mathcal{M}}_{\gamma_n/2}}}=\phi_{\Sigma_{\mathcal{M}_{\gamma_n/2}}} = \phi_0$. Therefore, we have:
\begin{equation}
\phi_{\Sigma_{\hat{\mathcal{M}}_{\gamma_n}}}=\min_{S \subseteq \{1,\ldots,d_n\}, |S| \leq s_n} \min_{\bm{v}\neq 0, |\bm{v}_{S^c}| \leq 3|\bm{v}_{S}|} \frac{\bm{v}^{T}\Sigma_{\hat{\mathcal{M}}_{\gamma_n}}\bm{v}}{\bm{v}^{T}\bm{v}} \geq \phi_0 \end{equation}

Using this along with Lemma 1 in \cite{MM2016} and Condition J, we have that DESIGN 3a is satisfied with $\phi_{\min}=\phi_0/16$, where $\inf_{v^Tv=1}v^T\Sigma_{11}v >2\phi_{\min}>0$, and $\Sigma_{11} $ is the covariance matrix of the relevant predictors. On the set $\mathcal{D}_n$, by Conditions K and L in our work, and Lemma 2 and proposition 1 in  \cite{MM2016}, assumption WEIGHTS is satisfied. On the set $\mathcal{A}_n \cap \mathcal{B}_n$, DGP 3 and DESIGN 3b are satisfied, while DESIGN 2 is satisfied by Condition L. 

Now by proposition 2, Lemmas 7 and 8 in \cite{MM2016} we obtain:
\begin{equation} P(sgn( \hat{\bm{\beta}}_{\hat{\mathcal{M}}_{\gamma_n}})=sgn(\bm{\beta})) \geq P(\mathcal{A}_n \cap \mathcal{B}_n \cap \mathcal{C}_n) \geq 1-P(\mathcal{A}_n^{\complement})-P(\mathcal{B}_n^{\complement})-P(\mathcal{C}_n^{\complement}) \end{equation}
$P(\mathcal{A}_n^{\complement})$ is given in Theorem \ref{theorem2} part i. For $P(\mathcal{B}_n^{\complement})$ using the method in the proof for Theorem \ref{theorem2}, we obtain:
\begin{equation} P(\mathcal{B}_n^{\complement}) \leq d_n^{'2}O\left(\frac{n^{\omega} K_{x,r}^{r}}{n^{r/2}}+ \exp\left(-n/K_{x,r}^4\right)\right) \end{equation}
And for $P(\mathcal{C}_n^{\complement})$:
\begin{equation} P(\mathcal{C}_n^{\complement}) \leq d_n'O\left(\frac{n^{\iota} K_{x,r}^{\tau}K_{\epsilon,q}^{\tau}}{\lambda_n^{\tau}n^{\tau\psi/2}}+ \exp\left(-\lambda_n^2 n^{\psi-1}/K^2_{x,r}K^2_{\epsilon,q}\right)\right) \end{equation}

To prove part ii) we follow the same steps from part i). We obtain $P(\mathcal{A}_n^{\complement})$, $P(\mathcal{B}_n^{\complement}), P(\mathcal{C}_n^{\complement})$ by following the method in the proof of Theorem \ref{theorem4}, and using Theorem 3 in \cite{WuandWu2016}.
\end{proof}

\subsection{Asymptotic Distribution of GLS estimator}\label{appendixB}
\begin{lem} Assume conditions E,F,G,H hold, then $\sqrt{n}(\hat{\beta}_k^{M}-\beta_k^{M})$ and $\sqrt{n}(\tilde{\beta}_k^{M}-\beta_k^{M})$ have the same asymptotic distribution.
\label{lem3}\end{lem}
\begin{proof} [Proof of Lemma \ref{lem3}]

It is clear that sufficient conditions for the feasible GLS estimator $\hat{\beta}_k^{M}$, and $\tilde{\beta}_k^{M}$ to have the same asymptotic distribution are \cite{Davidson2004}:
\begin{align*}
&\bm{X}_k^T(\hat{\Sigma}_{k,l_n}^{-1}-\Sigma_k^{-1})\bm{\epsilon}^k/\sqrt{n} \rightarrow 0 \\
&\bm{X}_k^T(\hat{\Sigma}_{k,l_n}^{-1}-\Sigma_k^{-1})\bm{X}_k/n \rightarrow 0
\end{align*}
By the proof of theorem \ref{theorem6}, both these conditions are satisfied, therefore $\hat{\beta}_k^{M}$, and $\tilde{\beta}_k^{M}$ have the same asymptotic distribution.
\end{proof}

We use the above lemma, and rely on the asymptotic distribution of $\tilde{\beta}_k^{M}$ to provide an explanation for the superior performance of GLSS, and its robustness to increasing levels of serial correlation in $\epsilon_{t,k}$. We deal with three cases, and we assume an AR(1) process for the errors for simplicity and ease of presentation. The results can be generalized to AR$(p)$ processes, by using the moving average representation of $\epsilon_{t,k}$:
\newline\newline
\underline{\textbf{Case 1:}} 
\newline\newline
We start with the setting used in figure \ref{figure1}, assume $x_{t,k}$ is iid and $\epsilon_{t,k}=\alpha\epsilon_{t-1,k}+e_t$, with $x_{t,k}$, and $\epsilon_{t,k}$ being independent $\forall t$. Using Gordin's central limit theorem \cite{Hayashi}, we calculate the asymptotic distribution of $\sqrt{n}(\tilde{\beta}_k^M-\beta_k^M) \rightarrow N(0,J)$, where $J=\frac{\sigma_{e}^2}{\sigma_{x_k}^2(1+\alpha^2)}$, $\sigma_{e}^2=var(e_t)$, and $\sigma_{x_k}^2=var(x_{t,k})$ . Using the same methods we calculate the asymptotic distribution of the marginal OLS estimator as $\sqrt{n}(\hat{\rho}_k-\rho_k) \rightarrow N(0,V)$, where $V=\frac{\sigma_{e}^2}{\sigma_{x_k}^2(1-\alpha^2)}$. Therefore the variance of the OLS estimator increases without bound as $\alpha$ increases towards 1. Whereas the variance of the GLS estimator actually decreases as $\alpha$ increases.
\newline\newline
\underline{\textbf{Case 2:}}
\newline\newline
We expand this to the case when $x_{t,k}$ is temporally dependent, for simplicity we let $x_{t,k}=\phi x_{t-1,k} + \eta_t$. We still assume $x_{j,k}$ and $\epsilon_t$ are independent $\forall j,t$, and $\epsilon_{t,k}=\alpha\epsilon_{t-1,k}+e_t$. This is the setting for the first model in the simulations section. Using Gordin's central limit theorem, and elementary calculations: $\sqrt{n}(\tilde{\beta}_k^M-\beta_k^M) \rightarrow N(0,J)$, where $J=\frac{(1-\phi^2)\sigma_{e}^2}{(1+\alpha^2-2\phi \alpha)\sigma_{\eta}^2}$. And for the marginal OLS estimator $\sqrt{n}(\hat{\rho}_k-\rho_k) \rightarrow N(0,V)$, where $V=\frac{(1+\phi^2)\sigma_e^2}{(1-\alpha^2)\sigma_{\eta}^2}$. We clearly see that for fixed $\phi$, the GLS estimate is robust to increasing $\alpha$, whereas the variance of the OLS estimator increases without bound as $\alpha$ increases towards 1. This sensitivity to $\alpha$ provides an explanation for the results seen in case 1 of the simulations, which show the performance of SIS severely deteriorates for high levels of serial correlation in $\epsilon_{t,k}$ 
\newline\newline
\underline{\textbf{Case 3:}}
\newline\newline
In both the previous cases, it is easy to see the GLS estimator is asymptotically efficient to the OLS estimator. For the case where $\bm{X}_k=(x_{t,k},t=1,\ldots,n)$ and $\bm{\epsilon}^k=(\epsilon_{t,k},t=1,\ldots,n)$ are dependent on each other, it is more complicated. In this setting, it is likely the case that $\rho_k \neq \beta_{k}^M$. Assume  $\epsilon_{t,k}=\alpha\epsilon_{t-1,k}+e_t$, and let $x_{t,k}-\alpha x_{t-1,k}=\tilde{x}_{t,k}$, and $W_1=\sum_{i=-\infty}^{\infty} \gamma(i)$, where $\gamma(i)=cov(\tilde{x}_{t,k}e_t,\tilde{x}_{t-i,k}e_{t-i})$. We start by examining the asymptotic distribution of $\sqrt{n}(\tilde{\beta}_k^M-\beta_k^M) \rightarrow N(0,J)$, where $J=W_1/(var(\tilde{x}_{t,k}))^2$. By the proof of theorem 1 in \cite{Wu2009}, $W_1 \leq (\sum_{t=0}^{\infty}\delta_{2}(\tilde{x}_{t,k}e_t))^2$, which gives us:
 \begin{align*}J\leq \frac{\left(\sum_{t=0}^{\infty}\delta_{2}(\tilde{x}_{t,k}\epsilon_t)\right)^2}{var(\tilde{x}_{t,k})^2} \leq \left(\frac{2||e_{t}||_4\Delta_{0,4}(\tilde{\bm{X}}_k)}{var(\tilde{x}_{t,k})}\right)^2 
 \end{align*}
 Where the last inequality follows from: $\delta_{2}(\tilde{x}_{t,k}e_t)=||e_0||_4||\tilde{x}_{t,k}-\tilde{x}_{t,k}^{*}||_4+||\tilde{x}_{0,k}||_4||e_t-e_t{*}||_4$. Since $e_t$ is iid $||e_t-e_t{*}||_4=0,\forall t>0$. If we assume, $x_{t,k}=\phi x_{t-1,k} + \eta_t$, by writing $\tilde{x}_{t,k}=\eta_t+(\phi-\alpha)x_{t-1,k}$, we have:
 \begin{align*} J \leq  \left(\frac{2||e_{t}||_4||\eta_t||_4|\phi-\alpha|}{(1-|\phi|)var(\tilde{x}_{t,k})}+\frac{2||e_{t}||_4||\eta_t||_4}{var(\tilde{x}_{t,k})}\right)^2 
 \end{align*}
 From these results we see that the asymptotic variance of the GLS estimator is bounded when $\alpha$ increases towards 1, and is largely robust to increasing levels of serial correlation in $\epsilon_{t,k}$. This result seems to provide an explanation for GLSS being robust to increasing levels of serial correlation in our simulations.
 
For the OLS estimator we obtain, $(\hat{\rho}_k-\rho_k) \rightarrow N(0,V)$, where $V=W_2/(var(x_{t,k}))^2$ and $W_2=\sum_{i=-\infty}^{\infty}cov(x_{t,k}\epsilon_t,x_{t-i}\epsilon_{t-i})$. As before, we can bound:
 \begin{align*}V \leq \frac{(\sum_{t=0}^{\infty}\delta_{2}(x_{t,k}e_t))^2}{var(x_{t,k}))^2} \leq \left(\frac{||\epsilon_{t,k}||_4\Delta_{0,4}(\bm{X}_k)}{var(x_{t,k})}+\frac{2||X_{k,t}||_4||e_t||_4}{(1-|\alpha|)var(x_{t,k})}\right)^2
 \end{align*}

We see the above bound is very sensitive to increasing serial correlation in $\epsilon_{t,k}$. Although this is an upper bound to the asymptotic variance, it seems to explain the deterioration in performance of SIS when increasing the serial correlation of $\epsilon_{t,k}$ in our simulations.
\section*{Acknowledgments}

The author would like to thank the two anonymous referees, the Associate Editor, and the Editor for their helpful
comments that have improved the paper.

\bibliographystyle{imsart-nameyear}
\bibliography{varselection}

\end{document}